\documentclass[11pt,reqno]{amsart}
\usepackage{tikz}
\usepackage{amsthm}
\usepackage{bm}
\usepackage{amsmath}
\usepackage{amsmath,amssymb}
\usepackage{subfig}
\usepackage{epstopdf}
\usepackage{epsfig}
\graphicspath{{./figures/}}
\usepackage{tikz}
\usetikzlibrary{shapes,arrows}
\tikzstyle{decision} = [diamond, draw, fill=blue!20, 
    text width=4.5em, text badly centered, node distance=3cm, inner sep=0pt]
\tikzstyle{block} = [rectangle, draw, fill=blue!20, 
    text width=5em, text centered, rounded corners, minimum height=4em]
\tikzstyle{line} = [draw, -latex']
\tikzstyle{cloud} = [draw, ellipse,fill=red!20, node distance=3cm,
    minimum height=2em]
 \newtheoremstyle{mystyle}{36pt}{}{}{}{\bfseries}{.}{ }{}
 
  \theoremstyle{plain}
 \newtheorem{theorem}{Theorem}
 \newtheorem{lemma}[]{Lemma}
 \newtheorem{definition}[]{Definition}
 
 \newtheorem{proposition}[]{Proposition}
 \newtheorem{example}{Example}

  \theoremstyle{remark}
 \newtheorem{remark}{Remark}

    \usetikzlibrary{positioning}
\tikzset{main node/.style={circle,fill=blue!20,draw,minimum size=1cm,inner sep=0pt},  }

\providecommand{\bbs}[1]{\left(#1\right)}

\newcommand{\ud}{\,\mathrm{d}}

\newcommand{\sP}{\mathcal{P}}

\newcommand{\ts}{\mathsf{T}}
\newcommand{\dd}{\mathcal{\dagger}}
\newcommand{\g}{\mathbf{g}}
\newcommand{\grad}{\mathrm{grad}}
\newcommand{\V}{\mathsf{V}}
\newcommand{\W}{\mathsf{W}}
\newcommand{\Hess}{\mathrm{Hess}}
\newcommand{\R}{\mathrm{R}}

\newcommand{\K}{\mathrm{K}}

\newcommand{\GG}{\Gamma^3}

\begin{document}
\title[Geometric calculations on Probability manifolds]{Geometric calculations on probability manifolds from reciprocal relations in Master equations}
\author[Li]{Wuchen Li}
\email{wuchen@mailbox.sc.edu}
\address{Department of Mathematics, University of South Carolina, Columbia, 29208.}
\keywords{Onsager reciprocal relations; Master equations with detailed balance; Thermodynamical probability manifolds; Curvatures in probability manifolds.}
\begin{abstract}
Onsager reciprocal relations model physical irreversible processes from complex systems. Recently, it has been shown that Onsager principles for master equations on finite states introduce a class of Riemannian metrics on the probability simplex, leading to probability manifolds or finite-state Wasserstein--2 spaces. In this paper, we study geometric calculations on probability manifolds, deriving the Levi-Civita connection, gradient, Hessian operators of energies, parallel transport, and calculating both the Riemannian and sectional curvatures. We present two examples of geometric quantities in probability manifolds. One example is the Levi-Civita connection from the chemical monomolecular triangle reaction. The other example is the sectional, Ricci, and scalar curvatures in Wasserstein space on a three-point lattice graph. 
\end{abstract}
\maketitle

\section{Introductions}
Non-equilibrium thermodynamics \cite{MFT, GEN} studies dynamical systems that interact with reservoirs for thermal energy or volume. Examples are coupled physical irreversible processes in complex systems, which are modeled by probability functions over physical states. Typical dynamics include thermoelectric phenomena and heat conduction in an anisotropic medium. In the modeling of irreversible processes, Onsager introduces the reciprocal relationship \cite{Onsager}, which describes the symmetric dissipation of the probability transition equations arising from physical processes, known as free energy dissipation or entropy production. Here, the free energy dissipation follows the second law of thermodynamics. Rather than addressing arbitrary physical mechanisms, one can focus on two canonical classes of stochastic models: overdamped Langevin dynamics in continuous state spaces, and detailed-balance Markov processes with transition rates on discrete domains. In both cases, entropy production is expressed as the time derivative of an entropy (or free energy) functional along the governing probability evolution equations, such as the Fokker--Planck equation in the continuous setting or the master equation in the discrete setting.

In recent decades, a particular type of gradient flow has been widely studied in optimal transport theory \cite{am2006, vil2008, LiYing}. They characterize a class of partial differential equations as gradient flows of energies on a probability space equipped with an infinite-dimensional Riemannian metric, namely the Wasserstein-$2$ space. Among these dynamics, the famous example is the probability density transition equation for overdamped Langevin dynamics, which is the gradient-drift Fokker-Planck equation. In this setting, the gradient drift Fokker-Planck equation satisfies the gradient flow of free energy in Wasserstein-$2$ space \cite{JKO,otto2001}. From the angle of non-equilibrium thermodynamics \cite{Ito}, the gradient drift Fokker-Planck equation satisfies the Onsager reciprocal relation. 
Along this direction, master equations with detailed balance conditions \cite{JS}  can also be formulated as gradient flows in discrete Wasserstein-$2$ spaces \cite{Chow2012, EM1, M} for probability distributions on discrete states. It also satisfies Onsager reciprocal relations \cite{Onsager}. From now on, we mainly study these generalized Wasserstein-$2$ spaces on a simplex, called {thermodynamical probability manifolds}, in short {\em probability manifolds}. 

Nowadays, geometric calculations in probability manifolds are essential for understanding fluctuation relations and thermodynamic properties. E.g., \cite{LiN1, WD} introduces overdamped Langevin dynamics in the probability manifold, namely Wasserstein common noises. It is based on geometric quantities, such as second-order operators in finite-state probability manifolds. Along this direction, a natural question arises. 
{\em What are geometric quantities, including Riemannian and sectional curvatures, in probability manifolds?}

In this paper, we study Riemannian calculations in a class of thermodynamical probability manifolds. This study is a generalization of Otto calculus on discrete states with general mobility functions. We first compute the Levi-Civita connection, gradient, and Hessian operators on probability manifolds. We further derive Riemannian and sectional curvatures in the probability simplex set. Two examples of these computations are performed. One example is the Levi-Civita connection on a three-state probability simplex, where the discrete probability manifold is constructed from the master equation for a chemical monomolecular triangle reaction. The other example is the Ricci curvature on a probability manifold with a three-point lattice graph. 

In literature, geometric calculations in Wasserstein-$2$ spaces have been studied in \cite{Lafferty, Lott}. More recently, \cite{LiHess, LiBreg, LiG1} generalized these calculations of metric spaces in both continuous and discrete states and studied divergence functions in Wasserstein-$2$  space. However, geometric calculations in generalized Wasserstein-$2$ metric spaces from the Onsager reciprocal relation on discrete states are still unknown. In this direction, \cite{LiG2, LiHess1} has studied the Hessian operators in generalized metric spaces on continuous and discrete state spaces (lattices), respectively. In this paper, we formulate geometric calculations in thermodynamical probability manifolds based on {\em generalized Onsager's response matrices}. We further compute Riemannian and sectional curvatures on probability manifolds. It is also worth mentioning that master equations and their Onsager reciprocal relations have been widely studied in the network theory \cite{JS}. Geometric calculations in probability manifolds also generalize this direction. The high-order derivative of Onsager response matrices defines the Riemannian curvature tensors in probability manifolds, shown in Theorem \ref{thm2}. We leave the related network analysis of curvature tensors for future work. 

In addition, there exist several studies connecting Onsager reciprocal relations \cite{Ito}, information geometry \cite{WS}, optimal transport on graphs \cite{Chow2012, EM1, M}, and dynamical density functional theory \cite{DDFT}. In particular, \cite{WS} investigates a class of Onsager response matrices constructed from the Hessian of entropy functionals and interprets Onsager gradient flows as natural gradient flows in information geometry. In the present work, the Onsager response matrix is given from the mean function $\theta$ arising in master equations and depends on the choice of free energy. From this perspective, intrinsic geometric properties of probability manifolds--such as connections and curvatures--can be related to Hessian and higher-order derivatives of the free energy, objects that are central in information geometry \cite{IG1, AyJostLeSchwachhoefer2017, Zhang2004}. See a concrete example in Example \ref{ex3}. 

The paper is organized as follows. In section \ref{sec2}, we briefly review reciprocal relations, master equations with detailed balance conditions, and optimal transport metric spaces on discrete states. In section \ref{sec3}, we introduce the Riemannian geometric calculations, including Levi-Civita connections, parallel transport equations, and curvature tensors. We also present two examples of geometric calculations on a three-point probability simplex in sections \ref{sec4} and \ref{sec5}.  

\section{Onsager reciprocal relations, network notations, and Thermodynamical probability manifolds}\label{sec2}
This section reviews finite-state master equations with detailed balance relations, Onsager reciprocal relations, network notation, and a class of thermodynamical probability manifolds. See more details in \cite{EM1, M, JS}.
\subsection{Master equations with detailed balance relations} 
Consider a discrete state $\big\{1,2,\cdots, n\big\}$. Denote a probability distribution $p(t)=(p_i(t))_{i=1}^n\in\mathbb{R}^n_+$ over states $i=1,2,\cdots, n$, that characterizes the discrete state system in a time domain $t\geq 0$, with
\begin{equation*}
0\leq p_i(t)\leq 1, \quad \sum_{i=1}^np_i(t)=1.
\end{equation*}
The master equation of the system $\big\{1,2,\cdots, n\big\}$ refers to the dynamical evolution of a probability function: 
\begin{equation}\label{master}
\frac{\ud p_i(t)}{\ud t} =\sum_{j=1}^n \Big[Q_{ji} p_j(t) - Q_{ij} p_i(t)\Big], 
\end{equation}
where there is an initial value probability function $p(0)$, and the nonnegative quantity $Q_{ji}\geq 0$, $1\leq i\neq j\leq n$, is the constant transition probability per time from state $j$ to state $i$. In modeling, this transition rate often represents the internal rate and external conditions, as determined by physical quantity differences between reservoir systems, such as temperature or pressure. 

Master equation \eqref{master} describes a continuous time Markov process on a finite state $\big\{1,2,\cdots ,n\big\}$ with a $Q$-matrix. Here, matrix $Q=(Q_{ij})_{1\leq i,j\leq n}$ is the generator of the continuous time Markov process. Denote $\pi=(\pi_i)_{i=1}^n\in \mathbb{R}^n$ with $\pi_i>0$ and $\sum_{i=1}^n\pi_i=1$ as a stationary point of equation \eqref{master}. We also call $\pi$ as the invariant distribution satisfying 
\begin{equation*}
\sum_{j=1}^n Q_{ji} \pi_j=\sum_{j=1}^nQ_{ij} \pi_i. 
\end{equation*}
A fact is that the master equation is with Onsager reciprocal relations, meaning that there exists a detailed balance condition for the Markov process with $Q$-matrix. See an example in section \ref{sec4}. More examples have been studied in \cite{JS}.   
\begin{definition}[Detailed balance condition]
Suppose that there exists a vector $\pi=(\pi_i)_{i=1}^n\in \mathbb{R}^n$, with $\pi_i>0$ and $\sum_{i=1}^n\pi_i=1$, such that 
\begin{equation}\label{db}
Q_{ij}\pi_i = Q_{ji} \pi_j, \quad \textrm{for $i$, $j\in\{1,2,\cdots, n\}$.}
\end{equation} 
\end{definition}
From now on, we study the master equation \eqref{master} with the detailed balance condition \eqref{db}. 
\subsection{Onsager reciprocal relations}
 One can rewrite equation \eqref{master} into a symmetric dissipation format \cite{Chow2012,EM1, M}, which satisfy Onsager reciprocal relations. This formulation represents gradient flows of free energies on probability manifolds. 

Denote a matrix $\omega=(\omega_{ij})_{1\leq i,j\leq n}\in\mathbb{R}^{n\times n}$, such that
\begin{equation}\label{omega}
\omega_{ij}=   Q_{ij}\pi_i\geq 0, \quad \textrm{for $i$, $j\in\{1,2,\cdots, n\}$.}
\end{equation}
From the detailed balance condition \eqref{db}, the weighted matrix $\omega$ is symmetric, since $\omega_{ij}=Q_{ij}\pi_i=Q_{ji}\pi_j=\omega_{ji}$. For a convex function $f\in C^1(\mathbb{R};\mathbb{R})$ with $f(1)=0$, master equation \eqref{master} can be written as below:
\begin{equation}\label{symmetrica}
 \begin{split}
   \frac{dp_i(t)}{dt}=&\sum_{j=1}^n \Big[Q_{ji}\pi_j \frac{p_j(t)}{\pi_j} - Q_{ij}\pi_i\frac{p_i(t)}{\pi_i}\Big]=\sum_{j=1}^n\omega_{ij}\Big(\frac{p_j(t)}{\pi_j}-\frac{p_i(t)}{\pi_i}\Big)\\
   =&\sum_{j=1}^n\omega_{ij}\frac{\frac{p_j(t)}{\pi_j}-\frac{p_i(t)}{\pi_i}}{f'(\frac{p_j(t)}{\pi_j})- f'(\frac{p_i(t)}{\pi_i})}\cdot\big(f'(\frac{p_j(t)}{\pi_j})- f'(\frac{p_i(t)}{\pi_i})\big)\\
   =&\sum_{j=1}^n \omega_{ij} \theta_{ij}(p(t)) \big(f'(\frac{p_j(t)}{\pi_j})- f'(\frac{p_i(t)}{\pi_i})\big). 
\end{split}
 \end{equation}
In equation \eqref{symmetrica}, we represent a matrix weight function $\theta\in C^1(\mathbb{R}^n; \mathbb{R}_+^{n\times n})$, with $\theta(p)=(\theta_{ij}(p))_{1\leq i,j\leq n}$ satisfying 
\begin{equation}\label{theta}
    \theta_{ij}(p):=\left\{\begin{aligned} 
    &\frac{\frac{p_j}{\pi_j}-\frac{p_i}{\pi_i}}{f'(\frac{p_j}{\pi_j})- f'(\frac{p_i}{\pi_i})}, &\textrm{if $\frac{p_i}{\pi_i}\neq \frac{p_j}{\pi_j}$;}\\
    & \frac{1}{f''(\frac{p_i}{\pi_i})} , &\quad\textrm{if $\frac{p_i}{\pi_i}=\frac{p_j}{\pi_j}$} .
    \end{aligned}\right.
    \end{equation}
In fact, the R.H.S. of symmetric formulation \eqref{symmetrica} can be written as a matrix function times a vector function. This is known as the strong Onsager gradient flow (physical terminology), which contains a symmetric matrix function (Onsager's response matrix) and the gradient of a free energy function (the force).  
To illustrate this point, we first define a force vector function. 
\begin{definition}[Force vector]
Define the free energy function, also named $f$--divergence, on the probability simplex: \begin{equation*}
 \mathrm{D}_{f}(p\|\pi) := \sum_{i=1}^n f(\frac{p_i}{\pi_i}) \pi_i.
 \end{equation*} 
Denote the Euclidean gradient operator of free energy by
\begin{equation*}
\nabla_p\mathrm{D}_{f}(p\|\pi)=\Big(\frac{\partial}{\partial p_i}\mathrm{D}_{f}(p\|\pi)\Big)_{i=1}^n= \Big(f'(\frac{p_i}{\pi_i})\Big)_{i=1}^n\in \mathbb{R}^n, 
\end{equation*}
where $\frac{\partial}{\partial p_i}$ is the Euclidean partial derivative with respect to $p_i$. 
In Onsager principle, one calls vector $\nabla_p\mathrm{D}_{f}(p\|\pi)$ a generalized force. 
\end{definition}  
 We further present the following symmetric matrix function, named Onsager's response matrix.  
\begin{definition}[Onsager's response matrix]
Denote a matrix function $L\in C^1(\mathbb{R}^{n\times n};\mathbb{R}^{n\times n})$. For a matrix $a\in\mathbb{R}^{n\times n}$, denote $L(a)=(L_{ij}(a))_{1\leq i,j\leq n}\in\mathbb{R}^{n\times n}$, such that
\begin{equation*}
L_{ij}(a):=\begin{cases}
-\omega_{ij}a_{ij},&\textrm{if $j\neq i$};\\
\sum_{k=1}^n\omega_{ki}a_{ki},&\textrm{if $j=i$.}
\end{cases}
\end{equation*}
The Onsager's response matrix refers to the composition of matrix functions $L\circ \theta\in C^1(\mathbb{R}^n; \mathbb{R}^{n\times n})$, such that $L(\theta(p)):=(L_{ij}(\theta(p)))_{1\leq i,j\leq n}$,
\begin{equation*}
L_{ij}(\theta(p)):=\begin{cases}
-\omega_{ij}\theta_{ij}(p),&\textrm{if $j\neq i$};\\
\sum_{k=1}^n\omega_{ki}\theta_{ki}(p),&\textrm{if $j=i$.}
\end{cases}
\end{equation*}
For simplicity of notation, we write $L(\theta)=L(\theta(p))$.  
We note that the matrix $L(\theta)$ is symmetric and semi-positive definite. I.e.,
\begin{equation*}
L_{ij}(\theta)=L_{ji}(\theta), 
\quad L(\theta)\succeq 0.
\end{equation*}
The semi-positive definiteness of $L(\theta)$ holds. For any vector $x=(x_i)_{i=1}^n\in \mathbb{R}^n$, we have \begin{equation*}
\begin{split}
x^{\ts}L(\theta)x=&\sum_{i=1}^n x_i \sum_{j=1}^n \omega_{ij}\theta_{ij}(p)(x_i-x_j)\\
=&\frac{1}{2}\sum_{i=1}^n\sum_{j=1}^n x_i  \omega_{ij}\theta_{ij}(p)(x_i-x_j)+\frac{1}{2}\sum_{i=1}^n\sum_{j=1}^n x_j \omega_{ji} \theta_{ji}(p)(x_j-x_i)\\
=&\frac{1}{2}\sum_{i=1}^n\sum_{j=1}^n \omega_{ij}(x_i-x_j)^2\theta_{ij}(p)\geq 0,
\end{split}
\end{equation*}
where we use the fact that $\theta_{ij}(p)=\theta_{ji}(p)\geq 0$, $\omega_{ij}=\omega_{ji}\geq 0$, for any $1\leq i\neq j\leq n$. 

\end{definition}
\begin{proposition}[Onsager reciprocal relations]
Master equation \eqref{master} can be rewritten as  
 \begin{equation} \label{o-gradient}
 \frac{\ud  p(t)}{\ud t} = - L(\theta(p(t))) \nabla_p \mathrm{D}_{f}(p(t)\|\pi), 
 \end{equation}
where $\nabla_p\mathrm{D}_{f}(p\|\pi)$ is the generalized force and $L(\theta)$ is the Onsager's response matrix. In addition, along with the solution of the master equation \eqref{master}, the free energy $\mathrm{D}_{f}(p(t)\|\pi)$ decays in the time variable. 
\begin{equation}\label{dissipation}
\frac{d}{dt}\mathrm{D}_{f}(p(t)\|\pi)=-\nabla_p \mathrm{D}_{f}(p(t)\|\pi)^{\ts}L(\theta(p(t)))\nabla_p \mathrm{D}_{f}(p(t)\|\pi)\leq 0. 
\end{equation}
In other words, $\mathrm{D}_{f}(p(t)\|\pi)$ is a Lyapunov function of master equation \eqref{master}. 
\end{proposition}
\begin{proof}
 From the definition of $L(\theta)$, we have 
\begin{equation*}
\begin{split}
 \frac{\ud  p_i(t)}{\ud t} =&\big(-L(\theta(p(t)))\nabla_p\mathrm{D}_{f}(p(t)\|\pi)\big)_i\\
 =&\sum_{j\neq i, j=1}^n L_{ij}(\theta(p(t))) \frac{\partial}{\partial p_j}\mathrm{D}_{f}(p(t)\|\pi)+ L_{ii}(\theta(p(t)))\frac{\partial}{\partial p_i}\mathrm{D}_{f}(p(t)\|\pi)\\
=&\sum_{j=1}^n \omega_{ij} \theta_{ij}(p(t)) \bbs{f'(\frac{p_j(t)}{\pi_j})- f'(\frac{p_i(t)}{\pi_i})}. 
\end{split}
\end{equation*}
Thus, equation \eqref{o-gradient} is a reformulation of the symmetric condition \eqref{symmetrica}. In addition, 
\begin{equation*}
\frac{d}{dt}\mathrm{D}_{f}(p(t)\|\pi)=\nabla_p\mathrm{D}_{f}(p(t)\|\pi)^{\ts}\frac{dp}{dt}=-\nabla_p \mathrm{D}_{f}(p(t)\|\pi)^{\ts}L(\theta(p(t)))\nabla_p \mathrm{D}_{f}(p(t)\|\pi)\leq 0, 
\end{equation*}
where we use the fact that $L(\theta)$ is semi-positive definite. This proves the result. 
\end{proof}
We present an important example of Onsager reciprocal relations, with a selected convex function $f$ and a weight function $\theta$. For example, we choose the $f$-divergence to be the KL divergence. 
\begin{example}[KL divergence mean \cite{EM1, M}]\label{exm2}
Consider $f(z)=z\log z-(z-1)$, where $\log$ is the natural logarithm function. Then the free energy is the f-divergence $\mathrm{D}_f(p\|\pi)$, also named Kullback--Leibler (KL) divergence: 
\begin{equation*}
 \mathrm{D}_{\mathrm{KL}}(p\|\pi)=\sum_{i=1}^np_i\log\frac{p_i}{\pi_i}. 
\end{equation*}
In this case, $\frac{\partial}{\partial p_i} \mathrm{D}_{\mathrm{KL}}(p\|\pi)=\log\frac{p_i}{\pi_i}+1$, and 
\begin{equation*}
    \theta_{ij}(p):= \frac{\frac{p_j}{\pi_j}-\frac{p_i}{\pi_i}}{\log\frac{p_j}{\pi_j}- \log\frac{p_i}{\pi_i}}.
    \end{equation*}
Thus, master equation \eqref{master} satisfies
\begin{equation*}
\frac{dp(t)}{dt}=-L(\theta(p(t)))\cdot\nabla_{p} \mathrm{D}_{\mathrm{KL}}(p(t)\|\pi).
\end{equation*}
\end{example}
We next study the generalization of KL divergence, namely $\alpha$-divergence \cite{IG1}. It introduces the following mean function to define the Onsager response matrix.
\begin{example}[$\alpha$-divergence mean]\label{exm2}
Consider $f(z)=\frac{4}{1-\alpha^2}(\frac{1-\alpha}{2}+\frac{1+\alpha}{2}z-z^{\frac{1+\alpha}{2}})$, where $\alpha\neq 1$. We call the f-divergence $\mathrm{D}_f(p\|\pi)$ as the $\alpha$-divergence function: 
\begin{equation*}
 \mathrm{D}_{\alpha}(p\|\pi)=\frac{4}{1-\alpha^2}\Big(1-\sum_{i=1}^n p_i^{\frac{1+\alpha}{2}}\pi_i^{\frac{1-\alpha}{2}}\Big). 
\end{equation*}
In this case, $\frac{\partial}{\partial p_i} \mathrm{D}_{\alpha}(p\|\pi)=\frac{2}{\alpha-1}(z^{\frac{\alpha-1}{2}}-1)$, and 
\begin{equation*}
    \theta_{ij}(p):= \frac{\alpha-1}{2}\frac{\frac{p_j}{\pi_j}-\frac{p_i}{\pi_i}}{(\frac{p_j}{\pi_j})^{\frac{\alpha-1}{2}}- (\frac{p_i}{\pi_i})^{\frac{\alpha-1}{2}}}.
    \end{equation*}
Thus, master equation \eqref{master} satisfies
\begin{equation*}
\frac{dp(t)}{dt}=-L(\theta(p(t)))\cdot\nabla_{p} \mathrm{D}_{\mathrm{KL}}(p(t)\|\pi).
\end{equation*}
\end{example}
\begin{example}[Geometric mean]
Given a reversible Markov chain with $Q$-matrix and invariant distribution $\pi\in\mathbb{R}_+^n$.  
Denote a geometric mean function 
\begin{equation*}
\theta_{ij}(p):=\sqrt{\frac{p_ip_j}{\pi_i\pi_j}}.
\end{equation*}
One can study a generalized Onsager gradient flow below. Define a free energy function $E\in C^{1}(\mathcal{P}_+;\mathbb{R})$. Consider a dynamical model by 
\begin{equation*}
\frac{d}{dt}p_i=-\Big(L(\theta)\nabla_p E(p)\Big)_i=\sum_{j\in N(i)}\sqrt{\frac{p_ip_j}{\pi_i\pi_j}}\omega_{ij}(\frac{\partial}{\partial p_j}E(p)-\frac{\partial}{\partial p_i}E(p)). 
\end{equation*}
\end{example}
\begin{remark}
We note that Onsager reciprocal relations \cite{Onsager} are not limited to equation \eqref{master}. These relations show that the macroscopic physical evolution equation can be written in the form \eqref{o-gradient} for general choices of weight functions $\theta$ and $f$-divergences (free energies). E.g., one can derive a general weight function $\theta$ from nonlinear master equations. On the other hand, many scientific computing methods also derive new classes of weight functions. Typical examples include the geometric mean, Harmonic mean, and arithmetic mean. See \cite{GLL,EM1, M}. These mean functions are also studied in modeling dynamical density functional theory \cite{DDFT}. 
The geometric calculations in this paper hold for general choices of average functions $\theta=(\theta_{ij})_{(i,j)\in E}$, under the assumption that $\theta_{ij}\in C^2(\mathcal{P}_+)$, for any indices $(i,j)\in E$.    
\end{remark}

\subsection{Network notations}
As in Kirchhoff's circuit laws \cite{JS}, we can apply graph gradient and divergence operators to represent Onsager reciprocal relations. Consider a weighted graph $G=(V, E, \omega)$, where
$V:=\{1,2,\cdots, n\}$ is the vertex set representing the physical system, and
$E:=\{(i,j), \,\, 1 \le i,j,\le n \colon \omega_{ij}>0 \}$ is the edge index set with weights $\omega_{ij}$. Denote the neighborhood set $N(i):= \{j\in V: (i,j) \in E\}$.

We review gradient, divergence, and Laplacian operators on graphs. Given a function $\Phi \colon V \to \mathbb{R}$, denote $\Phi=(\Phi_i)_{i=1}^n\in \mathbb{R}^n$. Define a weighted gradient as a function $\nabla_\omega \Phi \colon E \to \mathbb{R}$, 
\begin{equation*}
(i,j)  \, \mapsto \,\,  (\nabla_\omega\Phi)_{ij} :=\sqrt{\omega_{ij}}\,(\Phi_j-\Phi_i). 
\end{equation*} 
We call $\nabla_\omega\Phi$ a potential vector field on $E$. A general vector field is an antisymmetric function on $E$ such that 
$v=\big( v_{ij} \big)_{(i,j)\in E}$: 
\begin{equation*}
v_{ij}=-v_{ji}, \quad (i,j) \in E. 
\end{equation*}
The divergence of a vector field $v$ is a function $\mathrm{div}_\omega(v) \colon E \to \mathbb{R}$,
\begin{equation*}
i \, \mapsto \, \mathrm{div}_\omega(v)_i := \sum_{j\in N(i)}\sqrt{\omega_{ij}}\, v_{ij}.
\end{equation*}
For a function $\Phi$ on $V$, the weighted graph Laplacian $\mathrm{div}_\omega\circ\nabla_\omega \colon V \to \mathbb{R}$ satisfies 
\begin{equation*}
 i\, \mapsto \,
\mathrm{div}_\omega\bbs{\nabla_\omega\Phi }_i
=\sum_{j\in N(i)}\sqrt{\omega_{ij}}(\nabla_\omega\Phi)_{ij}=\sum_{j\in N(i)}\omega_{ij}\,(\Phi_j-\Phi_i).
\end{equation*}
We note that $\Delta_\omega\in\mathbb{R}^{n\times n}$ is a negative semi-definite matrix. 
Clearly, the graph divergence operator $\mathrm{div}_\omega$ is the adjoint of the gradient operator. In other words,  
\begin{equation*}
\begin{split}
\sum_{i=1}^n \Phi_i \mathrm{div}_\omega(v)_i=&\sum_{i=1}^n\sum_{j\in N(i)}\Phi_i\sqrt{\omega_{ij}}\, v_{ij}\\
=&\frac{1}{2}\Big[\sum_{i=1}^n\sum_{j\in N(i)}\Phi_i\sqrt{\omega_{ij}}\, v_{ij}+\sum_{j=1}^n\sum_{i\in N(j)}\Phi_j\sqrt{\omega_{ji}}\, v_{ji}\Big]\\
=&-\frac{1}{2}\sum_{(i,j)\in E} (\nabla_\omega \Phi)_{ij}v_{ij},
\end{split}
\end{equation*}
where we adopt the notation that $\sum_{(i,j)\in E}=\sum_{i=1}^n\sum_{j\in N(i)}=\sum_{j=1}^n\sum_{i\in N(j)}$, and use the fact that $v_{ij}=-v_{ji}$. 
The operator $\mathrm{div}_\omega\bbs{\nabla_\omega}\in\mathbb{R}^{n\times n}$ is negative definite, since for any vector $\Phi\in\mathbb{R}^n$, 
\begin{equation*}
\sum_{i=1}^n \Phi_i \mathrm{div}_\omega(\nabla_\omega \Phi)_i=-\frac{1}{2}\sum_{(i,j)\in E} (\nabla_\omega \Phi)^2_{ij}\leq 0.
\end{equation*}

We write the Onsager's response matrix as $L(\theta):=-\mathrm{div}_\omega(\theta \nabla_\omega)\colon V \to \mathbb{R}$. In other words, denote a vector $\Phi\in\mathbb{R}^n$,  we write $\theta(p)\nabla_{\omega}\Phi$ as a vector field, 
\begin{equation*}
(\theta(p)\nabla_\omega\Phi)_{ij}=\theta_{ij}(p)(\nabla_\omega\Phi)_{ij}, 
\end{equation*}
and 
 \begin{equation*}
 i\, \mapsto \,
\mathrm{div}_\omega\bbs{\theta(p)\nabla_\omega\Phi}_i=
\sum_{j\in N(i)} \omega_{ij}(\Phi_j-\Phi_i)\theta_{ij}(p). 
\end{equation*}
Under the above notation, the master equation \eqref{master} satisfies 
\begin{equation}\label{not1}
\frac{dp(t)}{dt}=\mathrm{div}_\omega\big(\theta(p(t)) \nabla_\omega \nabla_p\mathrm{D}_{f}(p(t)\|\pi)\big). 
\end{equation}
Thus, the time derivative of free energy in \eqref{dissipation} can be written as  
\begin{equation}\label{not2}
\begin{split}
\frac{d}{dt}\mathrm{D}_{f}(p(t)\|\pi)
=&-\frac{1}{2}\sum_{(i,j)\in E}\Big(\nabla_\omega\nabla_p\mathrm{D}_{f}(p(t)\|\pi)\Big)_{ij}^2\theta_{ij}(p(t))\leq 0,
\end{split}
\end{equation}
where 
\begin{equation*}
\Big(\nabla_\omega\nabla_p\mathrm{D}_{f}(p(t)\|\pi)\Big)_{ij}:=\sqrt{\omega_{ij}}(\frac{\partial}{\partial p_j}-\frac{\partial}{\partial p_i})\mathrm{D}_{f}(p(t)\|\pi). 
\end{equation*}
Equation \eqref{not1} is known as the Wasserstein gradient flow on finite state spaces \cite{am2006, EM1, M}. In addition, the dissipation of free energy in \eqref{not2} equals the squared norm in the probability simplex with optimal-transport-induced Riemannian metrics on finite states. We present these properties in the next section. 

\subsection{Thermodynamical probability manifolds}
We introduce probability manifolds based on Onsager's response matrices $L(\theta)$. We shall define gradient operators, arc-lengths of curves, and distance functions between two probabilities in the simplex set.  

Denote the probability simplex set without boundary by
\begin{equation*}
\sP_+:=\Big\{(p_i)_{i=1}^n\in \mathbb{R}^n \colon \sum_{i=1}^np_i=1, \quad p_i> 0\Big\}.
\end{equation*}
Denote the tangent space at $p\in\mathcal{P}_+$ by
\begin{equation*}
T_p\mathcal{P}_+ = \Big\{(\sigma_i)_{i=1}^n\in \mathbb{R}^n\colon  \sum_{i=1}^n\sigma_i=0 \Big\}.
\end{equation*}
We next apply $L(\theta)$ to define an inner product in the probability simplex. When $p\in \sP_+$, $L(\theta)$ is a symmetric semi-positive matrix, whose diagonalization satisfies  
\begin{equation}\label{diag}
L(\theta)=U(p)\begin{pmatrix}
0 & & &\\
& {\lambda_{1}(p)}& &\\
& & \ddots & \\
& & & {\lambda_{n-1}(p)}
\end{pmatrix}U(p)^{\ts},
\end{equation}
where $0<\lambda_1(p)\leq\cdots\leq \lambda_{n-1}(p)$ are eigenvalues of $L(\theta)$ in the ascending order, 
and $U(p)=(u_0(p),u_1(p),\cdots, u_{n-1}(p))\in \mathbb{R}^{n\times n}$ is the orthogonal matrix, whose column are corresponding eigenvectors, with $u_0=\frac{1}{\sqrt{n}}(1,\cdots, 1)^{\ts}$. Denote the pseudo-inverse matrix of $L(\theta)$ by
\begin{equation*}
R(\theta):=L(\theta)^{\dagger}=U(p)\begin{pmatrix}
0 & & &\\
& \frac{1}{\lambda_{1}(p)}& &\\
& & \ddots & \\
& & & \frac{1}{\lambda_{n-1}(p)}
\end{pmatrix}U(p)^{\ts},
\end{equation*}
where $\dagger$ represents the symbol for the pseudo-inverse operator of a matrix.

For any vector $\Phi\in\mathbb{R}^n$ up to a constant shift, we denote a tangent vector field in $T_p\mathcal{P}_+$ for any $p\in\mathcal{P}_+$ by
\begin{equation*}
\V_{\Phi}:=L(\theta)\Phi=-\mathrm{div}_\omega (\theta\nabla_\omega \Phi)\in T_p\mathcal{P}_+. 
\end{equation*}
We remark that the map $\Phi\rightarrow \V_{\Phi}$ is an isomorphism $\mathbb{R}^n/\mathbb{R}\rightarrow T_{ p}\mathcal{P}_+$. This is true from the diagonalization of $L(\theta)$ in \eqref{diag}, since there is only one constant eigenvector $u_0(p)$ with eigenvalue $0$ for any vector $p\in\sP_+$. 

\begin{definition}[Finite state Riemannian metric tensor]
Define the inner product $\g:\mathcal{P}_+\times T_p\mathcal{P}_+\times T_p\mathcal{P}_+\rightarrow\mathbb{R}$ by 
\begin{equation*}
\begin{split}
\g(p)(\V_{\Phi_1}, \V_{\Phi_2}):=\langle \V_{\Phi_1}, \V_{\Phi_2}\rangle(p):=\V^{\ts}_{\Phi_1}R(\theta)\V_{\Phi_2},
\end{split}
\end{equation*}
where $\Phi_k\in\mathbb{R}^n/\mathbb{R}$, $k=1,2$, where $\Phi_k$ is a vector in $\mathbb{R}^n$ up to a constant vector shrift in the direction of $u_0$, such that 
\begin{equation*}
\V_{\Phi_k}=L(\theta)\Phi_k=-\mathrm{div}_\omega (\theta\nabla_\omega \Phi_k)\in T_p\mathcal{P}_+. 
\end{equation*}
From the definition of pseudo-inverse matrix, $L(\theta)R(\theta)L(\theta)=L(\theta)$. Then  
\begin{equation*}
\begin{split}
\langle \V_{\Phi_1}, \V_{\Phi_2}\rangle(p)=&\Phi_1^{\ts}L(\theta)R(\theta)L(\theta)\Phi_2=\Phi_1^{\ts}L(\theta)\Phi_2\\
=&\frac{1}{2}\sum_{(i,j)\in E} (\nabla_\omega\Phi_1)_{ij}(\nabla_\omega\Phi_2)_{ij}\theta_{ij}(p).
\end{split}
\end{equation*}
\end{definition}
In literature \cite{EM1, M}, the inner product $\g$ defines a Wasserstein-$2$  type metric on the simplex set $\mathcal{P}_+$. 
 It is also a finite-dimensional Riemannian manifold. To emphasize its connection with the Onsager reciprocal relation, we call $(\mathcal{P}_+, \g)$ the {\em thermodynamical probability manifold}. 

We next illustrate some basic quantities in the manifold $(\mathcal{P}_+, \g)$. 
We first study the gradient operator. Mathematically, the Onsager reciprocal relation shows that master equation \eqref{master} is a gradient descent flow of $f$--divergence in $(\mathcal{P}_+, \g)$. Denote $\bar\grad=\grad_\g\colon \mathcal{P}_+\times C^\infty(\mathcal{P}_+; \mathbb{R})\rightarrow T_p\sP_+$. 
\begin{proposition}[Gradient operators]\label{prop2}
Denote an energy function $F\in C^\infty (\mathcal{P}_+; \mathbb{R})$. The gradient operator of function $F$ in $(\mathcal{P}_+, \g)$ satisfies  
\begin{equation*}
\bar\grad F( p)=L(\theta)\nabla_p F(p)=-\mathrm{div}_\omega(\theta \nabla_\omega\nabla_p F(p)),\end{equation*}
where 
\begin{equation*}
\Big(\mathrm{div}_\omega(\theta \nabla_\omega\nabla_p F(p))\Big)_i
=\sum_{j\in N(i)}w_{ij}(\frac{\partial}{\partial  p_j}-\frac{\partial}{\partial  p_i})F( p)\theta_{ij}( p).
\end{equation*}
In particular, if $F(p)=\mathrm{D}_{f}(p\|\pi)$, then the negative gradient direction of $f$-divergence satisfies the R.H.S. of master equation \eqref{master}. 
\begin{equation}\label{Onsager}
\begin{split}
-\bar\grad \mathrm{D}_{f}(p\|\pi)=&-L(\theta)\nabla_{p}\mathrm{D}_{f}(p\|\pi)\\
=&\mathrm{div}_\omega(\theta \nabla_\omega\nabla_p \mathrm{D}_{f}(p\|\pi))\\
=&\Big(\sum_{j=1}^n \Big[Q_{ji} p_j - Q_{ij} p_i\Big]\Big)_{i=1}^n.
\end{split}
\end{equation}
\end{proposition}
\begin{proof}
The proof follows from the definition of the gradient operator in a Riemannian manifold. Since $R(\theta)=L(\theta)^{\dagger}$ is a Riemannian metric tensor in $(\mathcal{P}_+, \g)$, then 
\begin{equation*}
R(\theta)\cdot\bar\grad F( p)=\nabla_p F(p)+c u_0, 
\end{equation*}
where $c\in\mathbb{R}$ is a constant. Thus 
\begin{equation*}
\bar \grad F( p)=R(\theta)^{\dagger}\nabla_p F(p)=(L(\theta)^{\dagger})^{\dagger}\nabla_pF(p)=L(\theta)\nabla_pF(p). 
\end{equation*}
When $F(p)=\mathrm{D}_{f}(p\|\pi)$, from the definition of $\bar\grad F$ and the R.H.S. of \eqref{symmetrica}, we derive equation \eqref{Onsager}.
\end{proof}
Onsager reciprocal relation introduces a class of Riemannian manifolds $(\mathcal{P}_+, \g)$. One can define functions related to the dynamical behaviors or properties of the master equation \eqref{master} or to general curves in probability simplex sets. 
E.g., one defines the arc length function in $(\mathcal{P}_+, \g)$.  
\begin{proposition}[Arc length function]
For a curve $\gamma\in C^{1}([0, T];\sP_+)$, where $T>0$, the arc length $\mathrm{Len}_\g(\gamma)$ of curve $\gamma$ is defined as   
\begin{equation}\label{L}
\mathrm{Len}_\g(\gamma):=\int_0^{T} \big|\dot\gamma(t)^{\ts}R(\theta(\gamma(t)))\dot\gamma(t)\big|^{\frac{1}{2}}dt, 
\end{equation}
where we denote $\dot\gamma(t)=\frac{d}{dt}\gamma(t)$. If we denote $\Phi\in C^1([0, T]; \mathbb{R}^n)$, such that $R(\theta(\gamma))\dot\gamma=\Phi$, i.e., $\dot\gamma=L(\theta(\gamma))\Phi$, then 
\begin{equation*}
\mathrm{Len}_\g(\gamma)=\int_0^{T} \big|\Phi(t)^{\ts}L(\theta(\gamma(t))\Phi(t))\big|^{\frac{1}{2}}dt.
\end{equation*}
\end{proposition}
Using arc length functions, one can further study the minimal arc length between two points in a simplex set. 
\begin{definition}[Minimal arc length problems and Distances]
Given two points $p^0$, $p^1\in\mathcal{P}_+$. The minimal arc length problem satisfies the following optimization problem: 
\begin{equation*}
\mathrm{Dist}(p^0,  p^1):=\inf_{\gamma} \Big\{\mathrm{Len}_\g(\gamma)\colon \gamma(0)=p^0,~~\gamma(1)=p^1\Big\},
\end{equation*}
where the minimization is taken over all continuous differentiable curves in the probability simplex set, $\gamma(t)\in C^1([0,1];\mathcal{P}_+)$, $t\in [0,1]$, connecting boundary points $\gamma(0)=p^0$, $\gamma(1)=p^1$. 
We note that the minimal value function $\mathrm{Dist}=\mathrm{Dist}_\g\colon \sP_+\times\sP_+\rightarrow\mathbb{R}$ represents the Wasserstein-$2$  distance on the simplex set. 
\end{definition}
In literature, the minimal arc length problem and its generalizations have been widely studied in \cite{GLL, EM1}. They belong to the study of dynamical optimal transport on discrete states. 

\section{Geometric calculations in Thermodynamical probability manifold}\label{sec3}
In this section, we are ready to present the main result of this paper. We derive Levi-Civita connections, parallel transport, geodesic curves, and curvature tensors in probability manifolds.  

\subsection{Levi-Civita connection, parallel transport and geodesics}
Given $\Phi\in \mathbb{R}^{n}$, we define a vector field $\V_{\Phi}:=L(\theta)\Phi\in {T}_ p \mathcal{P}_+$. Suppose $F\in C^{\infty}(\mathcal{P}_+; \mathbb{R})$. Denote the directional derivative of $F$ at direction $\V_\Phi$ as 
\begin{equation*}
\begin{split}
(\V_{\Phi}F)( p):=&\frac{d}{d\epsilon}|_{\epsilon=0}F( p+\epsilon L(\theta)\Phi)=\nabla_p F(p)^{\ts}L(\theta)\Phi.
\end{split}
\end{equation*}
We also define the first order directional derivative of matrix function $\theta$ at direction $\V_\Phi$ as $\V_\Phi\theta=((\V_\Phi\theta)_{ij})_{1\leq i,j\leq n}\in\mathbb{R}^{n\times n}$, such that
\begin{equation*}
(\V_\Phi\theta)_{ij}:=\frac{\partial\theta_{ij}}{\partial p_i}(\V_\Phi)_i+\frac{\partial\theta_{ij}}{\partial p_j}(\V_\Phi)_j.
\end{equation*}

We first compute commutators of two vector fields in $(\mathcal{P}_+, \g)$. Denote the commutator by $[\cdot, \cdot]\colon \sP_+\times T_p\sP_+\times T_p\sP_+\rightarrow T_p\sP_+$.
\begin{lemma}
Given vectors $\Phi_1$, $\Phi_2\in \mathbb{R}^{n}$, the commutator $[\V_{\Phi_1}, \V_{\Phi_2}]\in T_p\mathcal{P}_+$ satisfies 
\begin{equation}\label{comm}
[\V_{\Phi_1}, \V_{\Phi_2}]= L(\V_{\Phi_1}\theta)\Phi_2-L(V_{\Phi_2}\theta)\Phi_1.
\end{equation}
In details, denote $[\V_{\Phi_1}, \V_{\Phi_2}]=([\V_{\Phi_1}, \V_{\Phi_2}]_i)_{i=1}^n$, then 
\begin{equation*}
\begin{split}
[\V_{\Phi_1}, \V_{\Phi_2}]_i=&\quad\sum_{j\in N(i)}\sum_{i'\in N(i)}\frac{\partial\theta_{ij}}{\partial p_i}\theta_{ii'}\Big((\nabla_\omega\Phi_1)_{i'i}(\nabla_\omega\Phi_2)_{ji}-(\nabla_\omega\Phi_2)_{i'i}(\nabla_\omega\Phi_1)_{ji}\Big)\\
&+\sum_{j\in N(i)}\sum_{j'\in N(j)}
\frac{\partial\theta_{ij}}{\partial p_j}\theta_{jj'}\Big((\nabla_\omega\Phi_1)_{j'j} (\nabla_\omega\Phi_2)_{ji}-(\nabla_\omega\Phi_2)_{j'j} (\nabla_\omega\Phi_1)_{ji}\Big).
\end{split}
\end{equation*}
\end{lemma}
\begin{proof}
We note that 
\begin{equation*}
\begin{split}
([\V_{\Phi_1}, \V_{\Phi_2}])F( p)=&(\V_{\Phi_1}(\V_{\Phi_2}F))( p)-(\V_{\Phi_1}(\V_{\Phi_2}F))( p)\\
=&\frac{d}{d\epsilon}|_{\epsilon=0} (\V_{\Phi_2}F)( p+\epsilon L(\theta)\Phi_1)-\frac{d}{d\epsilon}|_{\epsilon=0} (\V_{\Phi_1}F)( p+\epsilon_2 L(\theta)\Phi_2)  \\
=&\V_{\Phi_1}(\nabla_ pF(p)^{\ts}L(\theta)\Phi_2)-\V_{\Phi_2}(\nabla_pF(p)^{\ts} L(\theta)\Phi_1)\\
=&\quad (L(\theta)\Phi_1)^{\ts}\nabla^2_{ pp}F(p)(L(\theta)\Phi_2)+\nabla_ pF(p)^{\ts}L(\V_{\Phi_1}\theta)\Phi_2\\
&-(L(\theta)\Phi_2)^{\ts}\nabla^2_{ pp}F(p)(L(\theta)\Phi_1)- \nabla_ pF(p)^{\ts}L(\V_{\Phi_2}\theta)\Phi_1  \\
=&\nabla_ pF(p)^{\ts}\Big(L(\V_{\Phi_1}\theta)\Phi_2- L(\V_{\Phi_2}\theta)\Phi_1\Big), 
\end{split}
\end{equation*}
where the last equality is based on the fact that $\nabla^2_{ pp}F$ is a symmetric matrix. In detail, 
\begin{equation*}
\begin{split}
&(L(\V_{\Phi_1}\theta)\Phi_2)_i
\\=&\sum_{j\in N(i)} (\V_{\Phi_1}\theta)_{ij} (\nabla_\omega\Phi_2)_{ji}\sqrt{\omega_{ij}}\\=&\sum_{j\in N(i)} (\frac{\partial\theta_{ij}}{\partial p_i}(\V_{\Phi_1})_i+\frac{\partial\theta_{ij}}{\partial p_j}(\V_{\Phi_1})_j)(\nabla_\omega\Phi_2)_{ji}\sqrt{\omega_{ij}}\\
=&\sum_{j\in N(i)}\Big(\sum_{i'\in N(i)}\frac{\partial\theta_{ij}}{\partial p_i}\theta_{ii'}(\nabla_\omega\Phi_1)_{i'i}(\nabla_\omega\Phi_2)_{ji}+\sum_{j'\in N(j)}
\frac{\partial\theta_{ij}}{\partial p_j}\theta_{jj'}(\nabla_\omega\Phi_1)_{j'j} (\nabla_\omega\Phi_2)_{ji}\Big)\sqrt{\omega_{ij}}.
\end{split}
\end{equation*}
Thus, by switching the indices $1$ and $2$ and subtracting them, we finish the proof. 
\end{proof}

We next derive the Levi-Civita connection in $(\mathcal{P}_+, \g)$. We need the following definition. 
\begin{definition}
For any $p\in \mathcal{P}_+$, define $\Gamma\colon \mathbb{R}^n\times\mathbb{R}^n\times \sP_+\rightarrow\mathbb{R}^n$, such that $\Gamma(\Phi_1, \Phi_2,p)=(\Gamma(\Phi_1,\Phi_2,p)_i)_{i=1}^n\in \mathbb{R}^n$, with  
\begin{equation*}
\Gamma(\Phi_1, \Phi_2, p)_i:=\sum_{j\in N(i)}(\nabla_\omega\Phi_1)_{ij}(\nabla_\omega\Phi_2)_{ij}\frac{\partial}{\partial p_i}\theta_{ij}(p).
\end{equation*}
\end{definition}
Denote the Levi-Civita connection by $\bar \nabla=\nabla^{\g}\colon \sP_+\times T_p\sP_+\times T_p\sP_+\rightarrow T_p\sP_+$. 
\begin{lemma}
The Levi-Civita connection $\bar \nabla$ in $(\sP_+, \g)$ satisfies
\begin{equation}\label{V12}
\bar\nabla_{\V_{\Phi_1}}\V_{\Phi_2}=\frac{1}{2}\Big(L(\V_{\Phi_1}\theta)\Phi_2-L(\V_{\Phi_2}\theta)\Phi_1+L(\theta)\Gamma(\Phi_1,\Phi_2,p)\Big).
\end{equation}
\end{lemma}
\begin{proof}
Since $\langle \V_{\Phi_2}, \V_{\Phi_3}\rangle( p)=\Phi_2^{\ts}L(\theta)\Phi_3$ is a $C^2$ function of $ p$, then 
\begin{equation}\label{V}
\begin{split}
\V_{\Phi_1}\langle \V_{\Phi_2}, \V_{\Phi_3}\rangle( p)=&\frac{d}{d\epsilon}|_{\epsilon=0} \Phi_2^{\ts}L( p+\epsilon L(\theta)\Phi_1)\Phi_3\\
=& \Phi_2^{\ts}L(\V_{\Phi_1}\theta)\Phi_3.
\end{split}
\end{equation}
From the definition of the Levi-Civita connection in the Koszul formula, we have
\begin{equation}\label{V11}
\begin{split}
2\langle\bar\nabla_{\V_{\Phi_1}}\V_{\Phi_2}, \V_{\Phi_3}\rangle=&\V_{\Phi_1}\langle \V_{\Phi_2}, \V_{\Phi_3}\rangle +\V_{\Phi_2}\langle \V_{\Phi_3}, \V_{\Phi_1}\rangle -\V_{\Phi_3}\langle \V_{\Phi_1}, \V_{\Phi_2}\rangle \\
+&\langle \V_{\Phi_3}, [\V_{\Phi_1}, \V_{\Phi_2}]\rangle-  \langle \V_{\Phi_2}, [\V_{\Phi_1}, \V_{\Phi_3}]\rangle-\langle \V_{\Phi_1}, [\V_{\Phi_2}, \V_{\Phi_3}]\rangle.
\end{split}
\end{equation}
Substituting \eqref{comm}, \eqref{V} into \eqref{V11}, we have 
\begin{equation}\label{V1}
\begin{split}
2\langle\bar\nabla_{\V_{\Phi_1}}\V_{\Phi_2}, \V_{\Phi_3}\rangle=&\quad\Phi_2^{\ts}L(\V_{\Phi_1}\theta)\Phi_3+\Phi_1^{\ts}L(\V_{\Phi_2}\theta)\Phi_3 -\Phi_1^{\ts}L(\V_{\Phi_3}\theta)\Phi_2 \\
&+\Phi_3^{\ts} L(\V_{\Phi_1}\theta)\Phi_2-\Phi_3^{\ts}L(\V_{\Phi_2}\theta)\Phi_1\\
&- \Phi_2^{\ts} L(\V_{\Phi_1}\theta)\Phi_3+\Phi_2^{\ts}L(\V_{\Phi_3}\theta)\Phi_1-\Phi_1^{\ts} L(\V_{\Phi_2}\theta)\Phi_3+\Phi_1^{\ts}L(\V_{\Phi_3}\theta)\Phi_2.
\end{split}
\end{equation}
By annealing equivalent formulas in \eqref{V1}, we derive
\begin{equation}\label{a}
\langle \bar\nabla_{\V_{\Phi_1}}\V_{\Phi_2}, \V_{\Phi_3}\rangle=\frac{1}{2}\Big(\Phi_3^{\ts}L(\V_{\Phi_1}\theta)\Phi_2- \Phi_3^{\ts}L(\V_{\Phi_2}\theta)\Phi_1\Big)+\frac{1}{2}\Phi_1^{\ts}L(\V_{\Phi_3}\theta)\Phi_2.
\end{equation}
\noindent\textbf{Claim:} 
\begin{equation}\label{claim}
\Phi_1^{\ts}L(\V_{\Phi_3}\theta)\Phi_2=\Phi_3^{\ts} L(\theta)\Gamma(\Phi_1, \Phi_2,p).
\end{equation}
\begin{proof}[Proof of Claim \eqref{claim}]
In fact, formula \eqref{claim} can be verified by the following facts
\begin{equation*}
\begin{split}
\Phi_1^{\ts}L(\V_{\Phi_3}\theta)\Phi_2
=&\frac{1}{2}\sum_{(i,j)\in {E}}(\nabla_\omega \Phi_1)_{ij}(\nabla_\omega \Phi_2)_{ij} (\V_{\Phi_3}\theta)_{ij}\\
=&\frac{1}{2}\sum_{(i,j)\in {E}}(\nabla_\omega \Phi_1)_{ij}(\nabla_\omega \Phi_2)_{ij}\Big[\frac{\partial\theta_{ij}}{\partial p_i}(L(\theta)\Phi_3)_i+ \frac{\partial\theta_{ij}}{\partial p_j}(L(\theta)\Phi_3)_j\Big]\\
=&\sum_{(i,j)\in {E}}(\nabla_\omega \Phi_1)_{ij}(\nabla_\omega \Phi_2)_{ij}\frac{\partial\theta_{ij}}{\partial p_i}(L(\theta)\Phi_3)_i\\
=&\sum_{i=1}^n\Big(\sum_{j\in N(i)}(\nabla_\omega \Phi_1)_{ij}(\nabla_\omega \Phi_2)_{ij}\frac{\partial\theta_{ij}}{\partial p_i}\Big)(L(\theta)\Phi_3)_i\\
=&\Phi_3^{\ts} L(\theta) \Gamma(\Phi_1, \Phi_2, p).
\end{split}
\end{equation*}
\end{proof}
Using the claim and substituting \eqref{claim} into \eqref{a}, we prove the result. 
 \end{proof}
\begin{lemma}\label{lemma5}
The following equality holds: 
\begin{equation*}
\bar\nabla_{\V_{\Phi_1}}\V_{\Phi_2}+\bar\nabla_{\V_{\Phi_2}}\V_{\Phi_1}=\V_{\Gamma(\Phi_1, \Phi_2, p)}.
\end{equation*}
\end{lemma}
\begin{proof}
Since \eqref{V12} holds and $$\bar\nabla_{\V_{\Phi_1}}\V_{\Phi_2}+ \bar\nabla_{\V_{\Phi_2}}\V_{\Phi_1}=L(\theta)\Gamma(\Phi_1, \Phi_2, p),$$
we prove the result.
\end{proof}

\begin{lemma}
The Levi-Civita connection coefficient at $p\in \sP_+$ is given as below. For any $\Phi_1$, $\Phi_2$, $\Phi_3\in \mathbb{R}^n$,   
\begin{equation}\label{LC}
\langle \bar\nabla_{\V_{\Phi_1}}\V_{\Phi_2}, \V_{\Phi_3}\rangle=\frac{1}{2}\Big\{\Phi_1^{\ts} L(\theta)\Gamma(\Phi_2, \Phi_3, p){-}\Phi_2^{\ts} L(\theta)\Gamma(\Phi_1, \Phi_3, p){+}\Phi_3^{\ts} L(\theta)\Gamma(\Phi_1, \Phi_2, p) \Big\}.
\end{equation}
In detail, 
\begin{equation}\label{iterative_Gamma}
\begin{split}
\langle \bar\nabla_{\V_{\Phi_1}}\V_{\Phi_2}, \V_{\Phi_3}\rangle=&\frac{1}{4}\sum_{(i,j)\in E} \Big\{(\nabla_\omega \Phi_1)_{ij}(\nabla_\omega \Gamma(\Phi_2, \Phi_3, p))_{ij}{-}(\nabla_\omega \Phi_2)_{ij}(\nabla_\omega \Gamma(\Phi_1, \Phi_3,p))_{ij}\\
&\hspace{1.6cm}{+}(\nabla_\omega \Phi_3)_{ij}(\nabla_\omega \Gamma(\Phi_1, \Phi_2,p))_{ij}\Big\}\theta_{ij}(p).
\end{split}
\end{equation}
\end{lemma}
\begin{proof}
The Levi-Civita connection \eqref{LC} is derived from the combination of equations \eqref{V1} and \eqref{claim}. In addition, to prove equation \eqref{iterative_Gamma}, we only need to write $\Phi_1^{\ts} L(\theta)\Gamma(\Phi_2, \Phi_3, p)$ in details. Note
\begin{equation*}
\begin{split}
\Phi_1^{\ts} L(\theta)\Gamma(\Phi_2, \Phi_3,p)=\frac{1}{2} \sum_{(i,j)\in E}(\nabla_\omega\Phi_1)_{ij}(\nabla_\omega \Gamma(\Phi_2, \Phi_3,p))_{ij}\theta_{ij}(p).
\end{split}
\end{equation*}
\end{proof}

We next compute the parallel transport. Let $ \gamma\colon [0,T]\rightarrow \mathcal{P}_+$ be a smooth curve, with a parameter $t>0$. 
Denote $\V_\Phi$ as the tangent direction of the curve $\gamma(t)$ at time $t$. I.e., $\frac{d \gamma(t)}{dt}=L(\theta(\gamma(t)))\Phi(t)=\V_{\Phi(t)}$. 
Consider a smooth vector field $\V_{\eta}$ given by $\eta(t)\in \mathbb{R}^n$. Then equation for $\V_{\eta}$ to be parallel along $\gamma(t)$ satisfies   
\begin{equation*}
\V_{\frac{d\eta}{dt}}+ \bar\nabla_{\V_{\Phi}}\V_\eta=0.
\end{equation*}

\begin{theorem}[Parallel transport in probability manifold]
For $\V_\eta$ to be parallel along the curve $\gamma$, then the following system of parallel transport equations holds: 
\begin{equation}\label{parallel}
\left\{\begin{aligned}
&\frac{d\gamma}{dt}-L(\theta)\Phi=0,\\
& L(\theta)\frac{d\eta}{dt}+\frac{1}{2}\Big(L(\V_\Phi\theta)\eta-L(\V_\eta\theta)\Phi+L(\theta)\Gamma(\Phi, \eta, \gamma)\Big)=0.
\end{aligned}\right.
\end{equation}
In addition, the following statements hold:
\begin{itemize}
\item[(i)] If $\eta_1(t)$, $\eta_2(t)$ is parallel along $\gamma(t)$, then  
\begin{equation*}
\frac{d}{dt}\langle \V_{\eta_1}, \V_{\eta_2}\rangle=0.
\end{equation*}
\item[(ii)] The geodesic equation satisfies  
\begin{equation}\label{geo}
\frac{d\gamma}{dt}-L(\theta)\Phi=0,\quad L(\theta)\Big(\frac{d\Phi}{dt}+ \frac{1}{2}\Gamma(\Phi, \Phi, \gamma)\Big)=0.
\end{equation}
In detail, a solution of geodesic equation satisfies a system of continuity equation and Hamilton-Jacobi equation on graphs:  
\begin{equation}\label{wgeo}
\left\{\begin{split}
& \frac{d\gamma_i}{dt}+\mathrm{div}_\omega(\theta(\gamma)\nabla_\omega \Phi)_i=0,\\
&\frac{d\Phi_i}{dt}+\frac{1}{2}\sum_{j\in N(i)}(\nabla_\omega\Phi)^2_{ij}\partial_i\theta_{ij}(\gamma)=0.
\end{split}\right.
\end{equation}
\end{itemize}
\end{theorem}
\begin{proof}
By the Levi-Civita connection \eqref{V12}, we derive \eqref{parallel}. (i). Since $\langle \V_{\eta_1}, \V_{\eta_2}\rangle=\eta_1^{\ts}L(\theta)\eta_2$, and $\eta_1$, $\eta_2$ satisfy \eqref{parallel}, then 
\begin{equation*}
\begin{split}
\frac{d}{dt}\langle \V_{\eta_1}, \V_{\eta_2}\rangle
=&(L(\theta)\frac{d\eta_1}{dt})^{\ts}\eta_2+\eta_1^{\ts}L(\V_{\Phi}\theta)\eta_2+\eta_1^{\ts}L(\theta)\frac{d\eta_2}{dt}\\
=&-\frac{1}{2}\eta_2^{\ts}L(\V_{\Phi}\theta)\eta_1+\frac{1}{2}\eta_2^{\ts}L(\V_{\eta_1}\theta)\Phi-\frac{1}{2}\Phi^{\ts}L(\V_{\eta_2}\theta)\eta_1\\
&+\eta_1^{\ts}L(\V_{\Phi}\theta)\eta_2-\frac{1}{2}\eta_1^{\ts}L(\V_{\Phi}\theta)\eta_2+\frac{1}{2}\eta_1^{\ts}L(\V_{\eta_2}\theta)\Phi-\frac{1}{2}\Phi^{\ts}L(\V_{\eta_1}\theta)\eta_2\\
=&0.
\end{split}
\end{equation*}
(ii). If $\eta=\Phi$, then $(\gamma, \eta)$ satisfies the geodesic equation:  
$$ \V_{\frac{d\Phi}{dt}}+\bar\nabla_{\V_{\Phi}}\V_{\Phi}=0.$$
Since $[\V_{\Phi}, \V_{\Phi}]=0$, then the geodesic equation satisfies
\begin{equation*}
L(\theta)(\frac{d\Phi}{dt}+\frac{1}{2}\Gamma(\Phi, \Phi, \gamma))=0.
\end{equation*}
We note that $L(\theta)u_0=0$, where $u_0=\frac{1}{\sqrt{n}}(1,\cdots, 1)^{\ts}\in\mathbb{R}^n$ is a constant vector. When $\gamma\in \mathcal{P}_+$, $u_0$ is the only basis for the kernel space of $L(\theta)$. 
Thus, there exists a constant vector $c(t)(1,\cdots, 1)^{\ts}$, where $c\colon [0,T]\rightarrow\mathbb{R}$ is a scalar function of the time variable $t$, such that 
\begin{equation*}
\frac{d\Phi}{dt}+\frac{1}{2}\Gamma(\Phi, \Phi, \gamma)=c(t)(1,\cdots, 1)^{\ts}.
\end{equation*}
Assume $c(t)$ is integrable in the time interval $[0, T]$, and we let $\hat\Phi_i(t)=\Phi_i(t)+\int_0^t c(s)ds$, for any $i\in \{1,\cdots, n\}$, then $(\gamma(t), \hat\Phi(t))$ satisfies the geodesic equation \eqref{wgeo}. This finishes the proof. 
\end{proof}

We last present the Hessian operators of smooth functions in $(\mathcal{P}_+, \g)$. 
\begin{lemma}\label{diff}
Given a function $F\in C^{2}(\sP_+; \mathbb{R})$, denote the Hessian operator of $F$ in $(\mathcal{P}_+, \g)$ as $\bar\Hess F:=\Hess_{\g}F\colon \sP_+\times\mathbb{R}^n\times \mathbb{R}^n\rightarrow\mathbb{R}$. 
Then the Hessian operator of $F$ at directions $\V_{\Phi_1}$, $\V_{\Phi_2}$ satisfies  
\begin{equation}\label{Hess}
\begin{split}
 \bar{\Hess}F( p)\langle \V_{\Phi_1}, \V_{\Phi_2}\rangle=&\quad\Phi_1^{\ts} L(\theta)\nabla^2_{pp}F( p)L(\theta)\Phi_2\\
 &+ \frac{1}{2}\nabla_{ p}F( p)^{\ts}\Big(L(\V_{\Phi_1}\theta)\Phi_2+L(\V_{\Phi_2}\theta)\Phi_1-L(\theta)\Gamma(\Phi_1, \Phi_2, p)\Big) . 
 \end{split}
 \end{equation}
In detail, 
\begin{equation}\label{Hess_explicit}
\begin{split}
& \bar{\Hess}F( p)\langle \V_{\Phi_1}, \V_{\Phi_2}\rangle\\=&\quad \frac{1}{4}\sum_{(i,j)\in E}\sum_{(k,l)\in E}(\nabla_\omega\nabla_p)_{ij}(\nabla_\omega\nabla_p)_{kl}F(p)(\nabla_\omega\Phi_1)_{ij}(\nabla_\omega\Phi_2)_{kl}\theta_{ij}(p)\theta_{kl}(p)\\
 &+ \frac{1}{4}\sum_{(i,j)\in E} \Big\{(\nabla_\omega \Phi_1)_{ij}(\nabla_\omega \Gamma(\Phi_2, \nabla_pF(p), p))_{ij}+(\nabla_\omega \Phi_2)_{ij}(\nabla_\omega \Gamma(\Phi_1, \nabla_pF(p),p))_{ij}\\
&\hspace{1.9cm}-(\nabla_\omega \nabla_p F(p))_{ij}(\nabla_\omega \Gamma(\Phi_1, \Phi_2,p))_{ij}\Big\}\theta_{ij}(p), 
 \end{split}
 \end{equation}
 where we denote
 \begin{equation*}
\begin{split}
(\nabla_\omega\nabla_p)_{ij}(\nabla_\omega\nabla_p)_{kl}F(p):=&\sqrt{\omega_{ij}}\sqrt{\omega_{kl}}(\frac{\partial}{\partial p_{j}}-\frac{\partial}{\partial p_{i}})(\frac{\partial}{\partial p_{l}}-\frac{\partial}{\partial p_{k}})F(p).
\end{split}
 \end{equation*}
 \end{lemma}
\begin{proof}
From the definition of Hessian operator in $(\mathcal{P}_+, \g)$, we have 
\begin{equation*}
\begin{split}
&\bar{\Hess}F( p)\langle \V_{\Phi_1}, \V_{\Phi_2}\rangle\\
=&\V_{\Phi_1}\langle \V_{\bar\grad F}, \V_{\Phi_2}\rangle- \langle \V_{\bar\grad F}, \bar\nabla_{\V_{\Phi_1}}\V_{\Phi_2}\rangle\\
=&\Phi_1^{\ts} L(\theta)\nabla^2_{pp}F( p)L(\theta)\Phi_2+ \nabla_pF(p)^{\ts}\Big(L(\V_{\Phi_1}\theta)\Phi_2-\bar\nabla_{\V_{\Phi_1}}\V_{\Phi_2}\Big)\\
=&\Phi_1^{\ts} L(\theta)\nabla^2_{pp}F( p)L(\theta)\Phi_2+\frac{1}{2}\nabla_pF(p)^{\ts}\Big(L(\V_{\Phi_1}\theta)\Phi_2+L(\V_{\Phi_2}\theta)\Phi_1-L(\theta)\Gamma(\Phi_1,\Phi_2,p)\Big).
\end{split}
\end{equation*}
From equation \eqref{V12} for the Levi-Civita connection, we derive the Hessian operator formula \eqref{Hess}. We next estimate two terms in the Hessian formula \eqref{Hess} to obtain equation \eqref{Hess_explicit}. Firstly, 
\begin{equation*}
\begin{split}
&\Phi_1^{\ts} L(\theta)\nabla^2_{pp}F( p)L(\theta)\Phi_2\\=&\sum_{i=1}^n\sum_{k=1}^n\frac{\partial^2}{\partial p_i\partial p_k}F(p)(L(\theta)\Phi_1)_{i}(L(\theta)\Phi_2)_{k}\\
=&\sum_{i=1}^n\sum_{k=1}^n\frac{\partial^2}{\partial p_i\partial p_k}F(p)(-\sum_{j\in N(i)}\theta_{ij}\sqrt{\omega}_{ij}(\nabla_{\omega}\Phi_1)_{ij})(-\sum_{l\in N(k)}\sqrt{\omega_{kl}}(\nabla_\omega\Phi_2)_{kl}\theta_{kl})\\
=&\frac{1}{4}\sum_{(i,j)\in E}\sum_{(k,l)\in E}\sqrt{\omega_{ij}}\sqrt{\omega_{kl}}\Big([\frac{\partial^2}{\partial p_i\partial p_k}-\frac{\partial^2}{\partial p_j\partial p_k}]-[\frac{\partial^2}{\partial p_i\partial p_l}-\frac{\partial^2}{\partial p_j\partial p_l}]\Big)F(p) \\
&\hspace{5cm}\cdot(\nabla_{\omega}\Phi_1)_{ij}(\nabla_\omega\Phi_2)_{kl}\theta_{ij}\theta_{kl}.
\end{split}
\end{equation*}
Secondly, from equations \eqref{V12} and \eqref{iterative_Gamma}, we prove the explicit formula for the term with $\nabla_p F(p)$. This finishes the proof.
 \end{proof}

\subsection{Riemannian curvature tensor}
We compute the Riemannian curvature tensor in $(\mathcal{P}_+, \g)$. Denote $\bar{\R}=\R_\g\colon \mathcal{P}_+\times \mathbb{R}^n/\mathbb{R}\times\mathbb{R}^n/\mathbb{R}\times\mathbb{R}^n/\mathbb{R}\rightarrow\mathbb{R}^n/\mathbb{R}$.

The following notations are needed. For $\Phi_1$, $\Phi_2\in\mathbb{R}^n$, define the second order direction derivative of matrix function $\theta$ at directions $\V_{\Phi_1}$, $\V_{\Phi_2}$, by $\W_{\Phi_1, \Phi_2}\theta=((\W_{\Phi_1,\Phi_2}\theta)_{ij})_{1\leq i,j\leq n}\in\mathbb{R}^{n\times n}$, such that
\begin{equation}\label{W12}
\begin{split}
&(\W_{\Phi_1, \Phi_2}\theta)_{ij}:=\V_{\Phi_2}(\frac{\partial\theta_{ij}}{\partial p_i})(\V_{\Phi_1})_i+\V_{\Phi_2}(\frac{\partial\theta_{ij}}{\partial p_j})(\V_{\Phi_1})_j . 
\end{split}
\end{equation}
Define $\nabla_p \theta L(\V_{\Phi_1}\theta)\Phi_2=((\nabla_p \theta L(\V_{\Phi_1}\theta)\Phi_2)_{ij})_{1\leq i,j\leq n}\in\mathbb{R}^{n\times n}$, such that 
\begin{equation*}
(\nabla_p \theta L(\V_{\Phi_1}\theta)\Phi_2)_{ij}:=\frac{1}{2}\Big[\frac{\partial \theta_{ij}}{\partial p_i}(L(\V_{\Phi_1}\theta)\Phi_2)_i+\frac{\partial \theta_{ij}}{\partial p_j}(L(\V_{\Phi_1}\theta)\Phi_2)_j\Big].
\end{equation*}
We denote $m(\Phi_1, \Phi_2)=(m(\Phi_1, \Phi_2)_{ij})_{1\leq i,j\leq n}\in\mathbb{R}^{n\times n}$, such that
\begin{equation*}
m(\Phi_1, \Phi_2)_{ij}:=-2(\W_{\Phi_1, \Phi_2}\theta)_{ij}-\Big(\nabla_p \theta L(\V_{\Phi_1}\theta)\Phi_2+\nabla_p\theta L(\V_{\Phi_2}\theta)\Phi_1 \Big)_{ij}.
\end{equation*}

\begin{theorem}[Rimennain curvature in probability manifold]\label{thm2}
Given potentials $\Phi_1$, $\Phi_2$, $\Phi_3$, $\Phi_4\in \mathbb{R}^{n}/\mathbb{R}$, the Riemannian curvature at directions $\V_{\Phi_1}$, $\V_{\Phi_2}$, $\V_{\Phi_3}$, $\V_{\Phi_4}$ satisfies 
\begin{equation}\label{tensor}
\begin{split}
&\langle \bar{\R}(\V_{\Phi_1}, \V_{\Phi_2})\V_{\Phi_3}, \V_{\Phi_4}\rangle\\
=&\frac{1}{4}\Big\{\Phi_2^{\ts} L(m(\Phi_1, \Phi_3))\Phi_4 + \Phi_1^{\ts}L(m(\Phi_2, \Phi_4))\Phi_3-\Phi_2^{\ts} L(m(\Phi_1, \Phi_4))\Phi_3- \Phi_1^{\ts} L(m(\Phi_2, \Phi_3))\Phi_4\\
&\quad+\Gamma(\Phi_1, \Phi_3, p)^{\ts}L(\theta) \Gamma(\Phi_2, \Phi_4, p)- \Gamma(\Phi_2, \Phi_3, p)^{\ts}L(\theta)\Gamma(\Phi_1, \Phi_4, p)\\
&\quad+[\V_{\Phi_1}, \V_{\Phi_3}]^{\ts}R(\theta)[\V_{\Phi_2}, \V_{\Phi_4}]-[\V_{\Phi_2}, \V_{\Phi_3}]^{\ts}R(\theta)[\V_{\Phi_1}, \V_{\Phi_4}]+2[\V_{\Phi_3}, \V_{\Phi_4}]^{\ts}R(\theta)[\V_{\Phi_1}, \V_{\Phi_2}]\Big\}.
\end{split}
\end{equation}
\end{theorem}
\begin{proof}
To derive the curvature tensor of vector fields $\V_{\Phi_a}$, $a=1, 2,3,4$,  we apply the following formula:
\begin{equation}\label{main}
\begin{split}
\langle \bar\R(\V_{\Phi_1}, \V_{\Phi_2}) \V_{\Phi_3}, \V_{\Phi_4}\rangle=& \langle \bar\nabla_{\V_{\Phi_1}}\bar\nabla_{\V_{\Phi_2}}\V_{\Phi_3}-\bar\nabla_{\V_{\Phi_2}}\bar\nabla_{\V_{\Phi_1}}\bar\V_{\Phi_3}-\bar\nabla_{[\V_{\Phi_1}, \V_{\Phi_2}]}\V_{\Phi_3} , \V_{\Phi_4} \rangle\\
=&\quad \V_{\Phi_1}\langle \bar\nabla_{\V_{\Phi_2}}\V_{\Phi_3}, \V_{\Phi_4}\rangle - \langle \bar\nabla_{\V_{\Phi_2}}\V_{\Phi_3}, \bar\nabla_{\V_{\Phi_1}}\V_{\Phi_4}\rangle \\
&- \V_{\Phi_2}\langle\bar\nabla_{\V_{\Phi_1}}\V_{\Phi_3},\V_{\Phi_4}\rangle+\langle\bar\nabla_{\V_{\Phi_1}}\V_{\Phi_3}, \bar\nabla_{\V_{\Phi_2}}\V_{\Phi_4}\rangle\\
&- \langle\bar\nabla_{[\V_{\Phi_1}, \V_{\Phi_2}]}\V_{\Phi_3}, \V_{\Phi_4}\rangle.
\end{split}
\end{equation}
We estimate the above formula in three steps. Firstly, from \eqref{a}, we denote
\begin{equation*}
\V_{abc}(p):=\langle \bar\nabla_{\V_{\Phi_a}}\V_{\Phi_b}, \V_{\Phi_c}\rangle =\frac{1}{2}\Big(\Phi_c^{\ts}L(\V_{\Phi_a}\theta)\Phi_b- \Phi_c^{\ts}L(\V_{\Phi_b}\theta)\Phi_a+\Phi_a^{\ts}L(\V_{\Phi_c}\theta)\Phi_b\Big).
\end{equation*}
Thus
\begin{equation}\label{ref1}
\begin{split}
\V_{\Phi_1}\langle \bar\nabla_{\V_{\Phi_2}}\V_{\Phi_3}, \V_{\Phi_4}\rangle
=&\frac{d}{d\epsilon}|_{\epsilon=0} \V_{234}(p+\epsilon L(\theta)\Phi_1) \\
=&\frac{1}{2}\Big(\Phi_4^{\ts}L(\V_{\Phi_1}\V_{\Phi_2}\theta)\Phi_3-\Phi_4^{\ts}L(\V_{\Phi_1}\V_{\Phi_3}\theta)\Phi_2+\Phi_2^{\ts}L(\V_{\Phi_1}\V_{\Phi_4}\theta)\Phi_3\Big).
\end{split}
\end{equation}
Similarly, by exchanging the index $1$, $2$, we have 
\begin{equation}\label{ref2}
\begin{split}
\V_{\Phi_2}\langle \bar\nabla_{\V_{\Phi_1}}\V_{\Phi_3}, \V_{\Phi_4}\rangle 
=&\frac{1}{2}\Big(\Phi_4^{\ts}L(\V_{\Phi_2}\V_{\Phi_1}\theta)\Phi_3-\Phi_4^{\ts}L(\V_{\Phi_2}\V_{\Phi_3}\theta)\Phi_1+\Phi_1^{\ts}L(\V_{\Phi_2}\V_{\Phi_4}\theta)\Phi_3\Big). 
\end{split}
\end{equation}

Secondly, since 
\begin{equation*}
\bar\nabla_{\V_{\Phi_2}}\V_{\Phi_3}=\frac{1}{2}[\V_{\Phi_2},\V_{\Phi_3}]+\frac{1}{2}L(\theta)\Gamma(\Phi_2, \Phi_3, p).
\end{equation*}
Then 
\begin{equation}\label{ref3}
\begin{split}
\langle \bar\nabla_{\V_{\Phi_2}}\V_{\Phi_3}, \bar\nabla_{\V_{\Phi_1}}\V_{\Phi_4}\rangle=&\bar\nabla_{\V_{\Phi_2}}\V_{\Phi_3}^{\ts}\cdot R(\theta)\cdot \bar\nabla_{\V_{\Phi_1}}\V_{\Phi_4}\\
=&\frac{1}{4}\Big([\V_{\Phi_2}, \V_{\Phi_3}]^{\ts}R(\theta)[\V_{\Phi_1}, \V_{\Phi_4}]+[\V_{\Phi_2}, \V_{\Phi_3}]^{\ts}\Gamma(\Phi_1,\Phi_4, p)\\
&\quad+[\V_{\Phi_1}, \V_{\Phi_4}]^{\ts}\Gamma(\Phi_2, \Phi_3, p)+\Gamma(\Phi_2,\Phi_3, p)^{\ts}L(\theta)\Gamma(\Phi_1, \Phi_4, p)\Big).
\end{split}
\end{equation}
Similarly, by exchanging the index $1$, $2$, we have \begin{equation}\label{ref4}
\begin{split}
\langle \bar\nabla_{\V_{\Phi_1}}\V_{\Phi_3}, \bar\nabla_{\V_{\Phi_2}}\V_{\Phi_4}\rangle
=&\frac{1}{4}\Big([\V_{\Phi_1}, \V_{\Phi_3}]^{\ts}R(\theta)[\V_{\Phi_2}, \V_{\Phi_4}]+[\V_{\Phi_1}, \V_{\Phi_3}]^{\ts}\Gamma(\Phi_2,\Phi_4, p)\\
&\quad+[\V_{\Phi_2}, \V_{\Phi_4}]^{\ts}\Gamma(\Phi_1, \Phi_3, p)+\frac{1}{4}\Gamma(\Phi_1,\Phi_3, p)^{\ts}L(\theta)\Gamma(\Phi_2, \Phi_4, p)\Big).
\end{split}
\end{equation}

Thirdly, denote $\V_{\Phi}=[\V_{\Phi_1}, \V_{\Phi_2}]$, where
$$\Phi=R(\theta)[\V_{\Phi_1}, \V_{\Phi_2}].$$
Then 
\begin{equation}\label{ref5}
\begin{split}
&\langle\bar\nabla_{[\V_{\Phi_1}, \V_{\Phi_2}]} \V_{\Phi_3}, \V_{\Phi_4}\rangle\\
=&\frac{1}{2}\Phi_4^{\ts}\Big\{L(\V_\Phi\theta)\Phi_3-L(\V_{\Phi_3}\theta)\Phi+L(\theta)\Gamma(\Phi, \Phi_3, p)\Big\}\\
=&\frac{1}{2}\Phi_4^{\ts}L(\V_\Phi\theta)\Phi_3-\frac{1}{2}\Phi_4^{\ts} L(\V_{\Phi_3}\theta)\Phi+\frac{1}{2}\Phi_3^{\ts}L(\V_{\Phi_4}\theta)\Phi\\
=&\frac{1}{2}\Phi_4^{\ts}L([\V_{\Phi_1}, \V_{\Phi_2}]\theta)\Phi_3+\frac{1}{2}[\V_{\Phi_4}, \V_{\Phi_3}]^{\ts}R(\theta)[\V_{\Phi_1}, \V_{\Phi_2}]\\
=&\frac{1}{2}\Phi_4^{\ts}L(\V_{\Phi_1}\V_{\Phi_2}\theta)\Phi_3-\frac{1}{2}\Phi_4^{\ts}L(\V_{\Phi_2}\V_{\Phi_1}\theta)\Phi_3+\frac{1}{2}[\V_{\Phi_4}, \V_{\Phi_3}]^{\ts}R(\theta)[\V_{\Phi_1}, \V_{\Phi_2}].
\end{split}
\end{equation}

For the simplicity of presentation, we denote 
\begin{equation}\label{abcd}
(abcd):=\Phi_a^{\ts}L(\V_{\Phi_b}\V_{\Phi_c}\theta)\Phi_d\ .
\end{equation}
Clearly, we have $(abcd)=(dbca)$. Substituting equations \eqref{ref1}, \eqref{ref2}, \eqref{ref3}, \eqref{ref4}, \eqref{ref5} into equation \eqref{main}, we have 
\begin{equation}\label{ref6}
\begin{split}
&\langle \bar\R(\V_{\Phi_1}, \V_{\Phi_2})\V_{\Phi_3}, \V_{\Phi_4}\rangle\\
=&\quad\frac{1}{4}\Big\{(2143)+(2413)\Big\}-\frac{1}{4}\Big\{(2134)+(2314)\Big\}+\frac{1}{4}\Big\{(1234)+(1324)\Big\}-\frac{1}{4}\Big\{(1243)+(1423)\Big\}\\ 
&+\frac{1}{4}  \Gamma(\Phi_1, \Phi_3, p)^{\ts}L(\theta) \Gamma(\Phi_2, \Phi_4, p)-\frac{1}{4}   \Gamma(\Phi_2, \Phi_3, p)^{\ts}L(\theta)\Gamma(\Phi_1, \Phi_4, p)\\
&+\frac{1}{4}\Big\{ [\V_{\Phi_1}, \V_{\Phi_3}]^{\ts}R(\theta)[\V_{\Phi_2}, \V_{\Phi_4}]- [\V_{\Phi_2}, \V_{\Phi_3}]^{\ts}R(\theta)[\V_{\Phi_1}, \V_{\Phi_4}]+2  [\V_{\Phi_3}, \V_{\Phi_4}]^{\ts}R(\theta)[\V_{\Phi_1}, \V_{\Phi_2}]\Big\}.
\end{split}
\end{equation}
We finish the derivation of equation \eqref{tensor}. 
\end{proof}
We present an explicit formulation of the Riemannian curvature $\bar\R$. Denote a third order iterative Gamma operator $\Gamma^3\colon \mathbb{R}^n/\mathbb{R}\times\mathbb{R}^n/\mathbb{R}\times\mathbb{R}^n/\mathbb{R}\times\mathbb{R}^n/\mathbb{R}\times\sP_+\rightarrow\mathbb{R}^{n\times n}$. Given vectors $\Phi_1$, $\Phi_2$, $\Phi_3$, $\Phi_4\in\mathbb{R}^n$, write 
$\GG(\Phi_1,\Phi_2,\Phi_3,\Phi_4,p)=(\GG(\Phi_1,\Phi_2,\Phi_3,\Phi_4,p)_{ij})_{1\leq i,j\leq n}$, such that
\begin{equation*}
\begin{split}
\GG(\Phi_1,\Phi_2,\Phi_3, \Phi_4, p)_{ij}:=&\frac{1}{2}\sum_{k=1}^n(\nabla_{\omega})_{ik}\Big((\nabla_{\omega}\Gamma(\Phi_1,\Phi_2,p))_{ij}(\nabla_\omega\Phi_4)_{ij}\frac{\partial\theta_{ij}}{\partial p_i}\Big)(\nabla_\omega \Phi_3)_{ik} \theta_{ik}.
\end{split}
\end{equation*}
Here, we denote a matrix $A=(A_{ij})_{i\leq i,j\leq n}\in\mathbb{R}^{n\times n}$, with $A_{ij}:=(\nabla_{\omega}\Gamma(\Phi_1,\Phi_2,p))_{ij}(\nabla_\omega\Phi_4)_{ij}\frac{\partial\theta_{ij}}{\partial p_i}$, and 
\begin{equation*}
(\nabla_\omega)_{ik}A_{ij}:=\sqrt{\omega_{ik}}\Big(A_{kj}-A_{ij}\Big). 
\end{equation*}
We also denote a matrix $C=(C_{ij})_{1\leq i,j\leq n}\in\mathbb{R}^{n\times n}$, such that 
\begin{equation*}
\begin{split}
(\nabla_\omega)_{ij}(\nabla_\omega)_{kl}C 
:=\sqrt{\omega_{ij}}\sqrt{\omega_{kl}}(C_{ik}-C_{il}-C_{jk}+C_{jl}), 
\end{split}
\end{equation*}
and write $\nabla_{\omega}\Phi_1\nabla_\omega\Phi_2\nabla_{pp}^2\theta=((\nabla_{\omega}\Phi_1\nabla_\omega\Phi_2\nabla_{pp}^2\theta)_{ij})_{1\leq i,j\leq n} \in \mathbb{R}^{n\times n}$, such that 
\begin{equation*}
(\nabla_{\omega}\Phi_1\nabla_\omega\Phi_2\nabla_{pp}^2\theta)_{ij}:=(\nabla_{\omega}\Phi_1)_{ij}(\nabla_\omega\Phi_2)_{ij}\frac{\partial^2\theta_{ij}}{\partial p_i\partial p_j}.
\end{equation*}

\begin{proposition}
The Riemannian curvature $\bar\R$ in equation \eqref{tensor} satisfies   
\begin{equation}\label{tensor_explicit}
\begin{split}
&\langle \bar \R(\V_{\Phi_1}, \V_{\Phi_2})\V_{\Phi_3}, \V_{\Phi_4}\rangle\\
=&\quad\frac{1}{2}\sum_{i,j,k,l=1}^n \frac{\partial^2\theta_{ij}}{\partial p_i\partial p_i} \theta_{ik}\theta_{il}\sqrt{\omega_{ik}}\sqrt{\omega_{il}}\Big\{-(\nabla_\omega\Phi_2)_{ij}(\nabla_\omega\Phi_4)_{ij}(\nabla_\omega \Phi_1)_{ik}(\nabla_\omega \Phi_3)_{il}\\
&\hspace{5.8cm}-(\nabla_\omega\Phi_1)_{ij}(\nabla_\omega\Phi_3)_{ij}(\nabla_\omega \Phi_2)_{ik}(\nabla_\omega \Phi_4)_{il}\\
&\hspace{5.8cm}+(\nabla_\omega\Phi_2)_{ij}(\nabla_\omega\Phi_3)_{ij}(\nabla_\omega \Phi_1)_{ik}(\nabla_\omega \Phi_4)_{il}\\
&\hspace{5.8cm}+(\nabla_\omega\Phi_1)_{ij}(\nabla_\omega\Phi_4)_{ij}(\nabla_\omega \Phi_2)_{ik}(\nabla_\omega \Phi_3)_{il}\Big\}\\
&+\frac{1}{8}\sum_{i,j,k,l=1}^n \theta_{ij}\theta_{kl}\Big\{-(\nabla_\omega)_{ij}(\nabla_\omega)_{kl}(\nabla_{\omega}\Phi_2\nabla_\omega\Phi_4\nabla_{pp}^2\theta) (\nabla_\omega \Phi_1)_{ij}(\nabla_\omega \Phi_3)_{kl}\\
&\hspace{3.1cm}-(\nabla_\omega)_{ij}(\nabla_\omega)_{kl}(\nabla_{\omega}\Phi_1\nabla_\omega\Phi_3\nabla_{pp}^2\theta) (\nabla_\omega \Phi_2)_{ij}(\nabla_\omega \Phi_4)_{kl}\\
&\hspace{3.1cm}+(\nabla_\omega)_{ij}(\nabla_\omega)_{kl}(\nabla_{\omega}\Phi_2\nabla_\omega\Phi_3\nabla_{pp}^2\theta) (\nabla_\omega \Phi_1)_{ij}(\nabla_\omega \Phi_4)_{kl}\\
&\hspace{3.1cm}+(\nabla_\omega)_{ij}(\nabla_\omega)_{kl}(\nabla_{\omega}\Phi_1\nabla_\omega\Phi_4\nabla_{pp}^2\theta) (\nabla_\omega \Phi_2)_{ij}(\nabla_\omega \Phi_3)_{kl}\Big\}\\
&+\frac{1}{8}\sum_{(i,j)\in E}\Big\{\quad-\GG(\Phi_2,\Phi_4, \Phi_1, \Phi_3, p)_{ij}-\GG(\Phi_2,\Phi_4, \Phi_3, \Phi_1, p)_{ij}\\
&\hspace{2.5cm}-\GG(\Phi_1,\Phi_3, \Phi_2, \Phi_4, p)_{ij}-\GG(\Phi_1,\Phi_3, \Phi_4, \Phi_2, p)_{ij}\\
&\hspace{2.5cm}+\GG(\Phi_2,\Phi_3, \Phi_1, \Phi_4, p)_{ij}+\GG(\Phi_2,\Phi_3, \Phi_4, \Phi_1, p)_{ij}\\
&\hspace{2.5cm}+\GG(\Phi_1,\Phi_4, \Phi_2, \Phi_3, p)_{ij}+\GG(\Phi_1,\Phi_4, \Phi_3, \Phi_2, p)_{ij}\\
&\hspace{1.0cm}+ \theta_{ij}\Big[(\nabla_\omega \Gamma(\Phi_1, \Phi_3, p))_{ij}(\nabla_\omega \Gamma(\Phi_2, \Phi_4, p))_{ij}-(\nabla_\omega \Gamma(\Phi_2, \Phi_3, p))_{ij}(\nabla_\omega \Gamma(\Phi_1, \Phi_4, p))_{ij}\Big]\Big\}\\
&+\frac{1}{4}\Big\{[\V_{\Phi_1}, \V_{\Phi_3}]^{\ts}R(\theta)[\V_{\Phi_2}, \V_{\Phi_4}]-[\V_{\Phi_2}, \V_{\Phi_3}]^{\ts}R(\theta)[\V_{\Phi_1}, \V_{\Phi_4}]+2[\V_{\Phi_3}, \V_{\Phi_4}]^{\ts}R(\theta)[\V_{\Phi_1}, \V_{\Phi_2}]\Big\}.
\end{split}
\end{equation}

\end{proposition}
\begin{proof}
We derive equation \eqref{tensor} explicitly. Note that 
\begin{equation}\label{tend}
\begin{split}
&\Phi_2^{\ts}L(m(\Phi_1, \Phi_3))\Phi_4\\=&\frac{1}{2}\sum_{(i,j)\in E} (\nabla_\omega\Phi_2)_{ij}(\nabla_\omega\Phi_4)_{ij}m(\Phi_1, \Phi_3)_{ij}\\
=& \frac{1}{2}\sum_{(i,j)\in E} (\nabla_\omega\Phi_2)_{ij}(\nabla_\omega\Phi_4)_{ij}\Big\{2\W_{\Phi_1, \Phi_3}\theta-\nabla_p \theta L(\V_{\Phi_1}\theta)\Phi_3-\nabla_p\theta L(\V_{\Phi_3}\theta)\Phi_1 \Big\}_{ij}.
\end{split}
\end{equation}
We next estimate the following three terms in equation \eqref{tend}. Firstly,  
\begin{equation}\label{formula}
\begin{split}
&\frac{1}{2}\sum_{(i,j)\in E} (\nabla_\omega\Phi_2)_{ij}(\nabla_\omega\Phi_4)_{ij}(2\W_{\Phi_1, \Phi_3}\theta)_{ij}\\
=&\sum_{(i,j)\in E} (\nabla_\omega\Phi_2)_{ij}(\nabla_\omega\Phi_4)_{ij}(\V_{\Phi_3}(\frac{\partial\theta_{ij}}{\partial p_i})(\V_{\Phi_1})_i+\V_{\Phi_3}(\frac{\partial\theta_{ij}}{\partial p_j})(\V_{\Phi_1})_j)\\
=&2\sum_{(i,j)\in E} (\nabla_\omega\Phi_2)_{ij}(\nabla_\omega\Phi_4)_{ij}(\V_{\Phi_3}\frac{\partial\theta_{ij}}{\partial p_i})(\V_{\Phi_1})_i\\
=&2\sum_{(i,j)\in E} (\nabla_\omega\Phi_2)_{ij}(\nabla_\omega\Phi_4)_{ij}\Big(\frac{\partial^2\theta_{ij}}{\partial p_i\partial p_i}(\V_{\Phi_1})_i(\V_{\Phi_3})_i+\frac{\partial^2\theta_{ij}}{\partial p_i\partial p_j}(\V_{\Phi_1})_i(\V_{\Phi_3})_j\Big)\\
=&\quad 2\sum_{i,j,k,l=1}^n (\nabla_\omega\Phi_2)_{ij}(\nabla_\omega\Phi_4)_{ij}\frac{\partial^2\theta_{ij}}{\partial p_i\partial p_i}\sqrt{\omega_{ik}}(\nabla_\omega \Phi_1)_{ik}\sqrt{\omega_{il}}(\nabla_\omega \Phi_3)_{il}\theta_{ik}\theta_{il}\\
&+\frac{1}{2}\sum_{i,j,k,l=1}^n (\nabla_\omega)_{ij}(\nabla_\omega)_{kl}(\nabla_{\omega}\Phi_2\nabla_\omega\Phi_4\nabla_{pp}^2\theta) (\nabla_\omega \Phi_1)_{ij}(\nabla_\omega \Phi_3)_{kl}\theta_{ij}\theta_{kl}.
\end{split}
\end{equation}
To derive the last equality of \eqref{formula}, we use the following derivation. Denote $A_{ij}=(\nabla_\omega\Phi_2)_{ij}(\nabla_\omega\Phi_4)_{ij}\frac{\partial^2\theta_{ij}}{\partial p_i\partial p_j}$, then 
\begin{equation*}
\begin{split}
\sum_{i,j=1}^n A_{ij}(\V_{\Phi_1})_i(\V_{\Phi_3})_j=&\sum_{i,j=1}^n A_{ij}\sum_{i'=1}^n\omega_{ii'}(\Phi_i-\Phi_i')\theta_{ii'}\sum_{j'=1}^n\omega_{jj'}(\Phi_j-\Phi_{j'})\theta_{jj'}\\
=&\frac{1}{2}\sum_{i,i',j,j'=1}^n (A_{ij}-A_{i'j})(\Phi_i-\Phi_i')(\Phi_j-\Phi_{j'})\omega_{ii'}\omega_{jj'}\theta_{ii'}\theta_{jj'}\\
=&\frac{1}{4}\sum_{i,i',j,j'=1}^n \Big(A_{ij}-A_{i'j}-A_{ij'}+A_{i'j'}\Big)(\Phi_i-\Phi_i')(\Phi_j-\Phi_{j'})\omega_{ii'}\omega_{jj'}\theta_{ii'}\theta_{jj'}.
\end{split}
\end{equation*}
This finishes the derivation by switching indices $(i,i'), (j,j')$ to $(i,j), (k,l)$.

Secondly, 
\begin{equation*}
\begin{split}
&\frac{1}{2}\sum_{(i,j)\in E}(\nabla_\omega\Phi_2)_{ij}(\nabla_\omega\Phi_4)_{ij}(\nabla_p \theta L(\V_{\Phi_1}\theta)\Phi_3)_{ij}\\
=&\frac{1}{4}\sum_{(i,j)\in E}(\nabla_\omega\Phi_2)_{ij}(\nabla_\omega\Phi_4)_{ij}\Big(\frac{\partial \theta_{ij}}{\partial p_i}(L(\V_{\Phi_1}\theta)\Phi_3)_i+\frac{\partial \theta_{ij}}{\partial p_j}(L(\V_{\Phi_1}\theta)\Phi_3)_j\Big)\\
=&\frac{1}{2}\sum_{(i,j)\in E}(\nabla_\omega\Phi_2)_{ij}(\nabla_\omega\Phi_4)_{ij}\frac{\partial \theta_{ij}}{\partial p_i}(L(\V_{\Phi_1}\theta)\Phi_3)_i\\
=&\frac{1}{2}\sum_{i=1}^n\Gamma(\Phi_2,\Phi_4,p)_i(L(\V_{\Phi_1}\theta)\Phi_3)_i\\
=&\frac{1}{4}\sum_{i,j=1}^n(\nabla_{\omega}\Gamma(\Phi_2,\Phi_4,p))_{ij}(\nabla_\omega\Phi_3)_{ij}(\V_{\Phi_1}\theta)_{ij}\\
=&\frac{1}{2}\sum_{i,j=1}^n(\nabla_{\omega}\Gamma(\Phi_2,\Phi_4,p))_{ij}(\nabla_\omega\Phi_3)_{ij}\frac{\partial\theta_{ij}}{\partial p_i} (L(\theta)\Phi_1)_i\\
=&\frac{1}{2}\sum_{i,j=1}^n(\nabla_{\omega}\Gamma(\Phi_2,\Phi_4,p))_{ij}(\nabla_\omega\Phi_3)_{ij}\frac{\partial\theta_{ij}}{\partial p_i}\sum_{k=1}^n\sqrt{\omega_{ik}}(\nabla_\omega \Phi_1)_{ik} \theta_{ik}\\
=&\frac{1}{4}\sum_{i,j,k=1}^n(\nabla_{\omega})_{ik}\Big((\nabla_{\omega}\Gamma(\Phi_2,\Phi_4,p))_{ij}(\nabla_\omega\Phi_3)_{ij}\frac{\partial\theta_{ij}}{\partial p_i}\Big)(\nabla_\omega \Phi_1)_{ik} \theta_{ik}\\
=&\sum_{(i,j)\in E}\frac{1}{2}\GG(\Phi_2,\Phi_4, \Phi_1, \Phi_3, p)_{ij},
\end{split}
\end{equation*}
where we apply the definition of $\GG$. Thirdly, we switch the indices $1$ and $3$ in the above formula. We last derive 
\begin{equation*}
\Gamma(\Phi_1, \Phi_3, p)^{\ts}L(\theta) \Gamma(\Phi_2, \Phi_4, p)\\
=\frac{1}{2}\sum_{(i,j)\in E} (\nabla_\omega \Gamma(\Phi_1, \Phi_3, p))_{ij}(\nabla_\omega \Gamma(\Phi_2, \Phi_4, p))_{ij}\theta_{ij}.
\end{equation*}
Combining the above derivations in equation \eqref{tensor}, we derive equation \eqref{tensor_explicit}. 
\end{proof}

\subsection{Sectional curvatures}
We last derive the sectional, Ricci, and scalar curvatures in $(\mathcal{P}_+, \g)$. Define the sectional curvature in $(\sP_+, \g)$ as $\bar\K:=\K^\g\colon \mathcal{P}_+\times\mathbb{R}^n/\mathbb{R}\times\mathbb{R}^n/\mathbb{R}\rightarrow\mathbb{R}$. 
\begin{proposition}
Given potential $\Phi_1$, $\Phi_2\in \mathbb{R}^{n}$, the sectional curvature at directions $\V_{\Phi_1}$, $\V_{\Phi_2}$ satisfies 
\begin{equation*}
\begin{split}
&\bar\K(\V_{\Phi_1},\V_{\Phi_2})\\
=&\frac{1}{4}\Big\{\Phi_2^{\ts} L(m(\Phi_1, \Phi_2))\Phi_1 + \Phi_1^{\ts}L(m(\Phi_2, \Phi_1))\Phi_2-\Phi_2^{\ts}L(m(\Phi_1, \Phi_1))\Phi_2- \Phi_1^{\ts}L(m(\Phi_2, \Phi_2))\Phi_1\\
&\quad+\Gamma(\Phi_1, \Phi_2, p)^{\ts}L(\theta) \Gamma(\Phi_1, \Phi_2, p)- \Gamma(\Phi_2, \Phi_2, p)^{\ts}L(\theta)\Gamma(\Phi_1, \Phi_1, p)\\
&\quad-3[\V_{\Phi_1}, \V_{\Phi_2}]^{\ts} R(\theta) [\V_{\Phi_1}, \V_{\Phi_2}]\Big\}\cdot\Big(\langle \V_{\Phi_1}, \V_{\Phi_1}\rangle\langle \V_{\Phi_2}, \V_{\Phi_2}\rangle- \langle \V_{\Phi_1}, \V_{\Phi_2}\rangle^2\Big)^{-1}. 
\end{split}
\end{equation*}
\end{proposition}
\begin{proof}
From the definition of sectional curvature, we have 
\begin{equation*}
{\bar\K}(\V_{\Phi_1},\V_{\Phi_2})=\frac{\bar{\R}((\V_{\Phi_1},\V_{\Phi_2})\V_{\Phi_2}, \V_{\Phi_1})}{\langle \V_{\Phi_1}, \V_{\Phi_1}\rangle\langle \V_{\Phi_2}, \V_{\Phi_2}\rangle- \langle \V_{\Phi_1}, \V_{\Phi_2}\rangle^2}. 
\end{equation*}
Using the fact that $[\V_{\Phi_a}, V_{\Phi_a}]=0$, $a=1,2$, we finish the proof. 
\end{proof}
  
\section{Example I: Probability manifolds from chemical monomolecular triangle reactions}\label{sec4}
In this section, we present an example of a thermodynamical probability manifold from a system of monomolecular triangle reactions. We note that this is a classic example of Onsager reciprocal relations \cite[Section 3]{Onsager}. 
We present the Levi-Civita connection for this probability manifold. 

\subsection{Reciprocal relations} Suppose there is a homogenous phase in three forms $\{A=1, B=2, C=3\}$. A phase $i\in \{1,2,3\}$ can transform itself into others as follows. 
\begin{center}
\begin{tikzpicture}[->,shorten >=1pt,auto,node distance=2cm,
        thick,main node/.style={circle,fill=blue!20,draw,minimum size=0.5cm,inner sep=0pt]} ]
   \node[main node] (1) {A};
    \node[main node] (2) [below right of=1]  {C};
    \node[main node](3)[below left of=1]{B};
    \path[<->]
    (1) edge node {} (2)
    (2) edge node{} (3)
    (1) edge node{} (3);
\end{tikzpicture}
\end{center}
Assume that the reaction obeys a simple mass-action law. The fraction of molecule $A=1$ changes into molecule $B=2$ in a short time $\Delta t>0$ follows 
$Q_{12}\Delta t$,
where $Q_{12}$ is a constant reaction rate. It means that the rates of change for the density of three species $A$, $B$, $C$, denoted by  $p_1$, $p_2$, $p_3$, satisfy the following master equation system: 
\begin{equation}\label{mc}
\left\{\begin{aligned}
& \frac{dp_1}{dt}=Q_{21}p_2+Q_{31}p_3-(Q_{12}+Q_{13})p_1,\\
& \frac{dp_2}{dt}=Q_{12}p_1+Q_{32}p_3-(Q_{21}+Q_{23})p_2,\\
&\frac{dp_3}{dt}=Q_{13}p_1+Q_{23}p_2-(Q_{31}+Q_{32})p_3. 
\end{aligned}\right.
\end{equation}
Denote $\pi_1$, $\pi_2$, $\pi_3$ as the equilibrium of the master equation \eqref{mc}. This means 
\begin{equation*}
Q_{21}\pi_2+Q_{31}\pi_3-(Q_{12}+Q_{13})\pi_1=0.
\end{equation*}
Similar conditions hold for the second and third equations in the master equation system \eqref{mc}. In this case, the Onsager reciprocal relation means that the following detailed balance condition holds. We require 
\begin{equation*}
Q_{12}\pi_1=Q_{21}\pi_2, \quad Q_{23}\pi_2=Q_{32}\pi_3, \quad Q_{31}\pi_3=Q_{13}\pi_1. 
\end{equation*}
Clearly, the detailed balance condition implies that 
\begin{equation*}
Q_{12}Q_{23}Q_{31}=Q_{21}Q_{32}Q_{13}. 
\end{equation*}
The above is known as the time-reversible property of the Markov chain. Denote a relative entropy function (a free energy up to a constant) by
\begin{equation*}
\mathrm{D}_{\mathrm{KL}}(p\|\pi)=p_1\log\frac{p_1}{\pi_1}+p_2\log\frac{p_2}{\pi_2}+p_3\log\frac{p_3}{\pi_3}.
\end{equation*}
The force vector is the differential of the free energy 
\begin{equation*}
d \mathrm{D}_{\mathrm{KL}}(p\|\pi)=  (\log\frac{p_1}{\pi_1}+1)d p_1+  (\log\frac{p_2}{\pi_2}+1)d p_2+  (\log\frac{p_3}{\pi_3}+1)d p_3. 
\end{equation*}
Denote the response matrix function by
\begin{equation*}
L(\theta)=\begin{pmatrix}
\omega_{12}\theta_{12}+\omega_{13}\theta_{13} & -\omega_{12}\theta_{12} & -\omega_{13}\theta_{13} \\
-\omega_{12}\theta_{12}   & \omega_{12}\theta_{12}+\omega_{23}\theta_{23} & -\omega_{23}\theta_{23}\\
-\omega_{13}\theta_{13}  & -\omega_{23}\theta_{23}   & \omega_{13}\theta_{13}+\omega_{23}\theta_{23}
\end{pmatrix},
\end{equation*}
where  
\begin{equation*}
\omega_{12}=Q_{12}\pi_1, \quad \omega_{23}=Q_{23}\pi_2, \quad \omega_{13}=Q_{13}\pi_1,  
\end{equation*}
and 
\begin{equation*}
\theta_{12}(p)=\frac{\frac{p_1}{\pi_1}-\frac{p_2}{\pi_2}}{\log\frac{p_1\pi_2}{p_2\pi_1}},\quad \theta_{23}(p)=\frac{\frac{p_2}{\pi_2}-\frac{p_3}{\pi_3}}{\log\frac{p_2\pi_3}{p_3\pi_2}}, \quad \theta_{13}(p)=\frac{\frac{p_1}{\pi_1}-\frac{p_3}{\pi_3}}{\log\frac{p_1\pi_3}{p_3\pi_1}}. 
\end{equation*}
Thus, $R(\theta)=L(\theta)^{\dd}$. Thus, equation \eqref{mc} can be written as the Onsager gradient flow: 
\begin{equation*}
\frac{dp}{dt}=-L(\theta)\cdot \nabla_p \mathrm{D}_{\mathrm{KL}}(p\|\pi).  
\end{equation*}

\subsection{Levi-Civita connection in probability manifolds}
Denote the probability simplex set on a three-state space as 
\begin{equation*}
\Delta_3:=\Big\{(p_1, p_2, p_3)\in\mathbb{R}^3\colon p_1+p_2+p_3=1,\quad p_1, p_2, p_3>0\Big\}.
\end{equation*}
Given a vector $\Phi\in \mathbb{R}^3$ and a point $p\in \Delta_3$, the metric $\g$ satisfies 
\begin{equation*}
\begin{split}
\langle \V_\Phi, \V_\Phi\rangle
=&(\nabla_\omega\Phi)^2_{12} \theta_{12}+(\nabla_\omega \Phi)_{23}^2\theta_{23}+(\nabla_\omega\Phi)_{13}^2\theta_{13},
\end{split}
\end{equation*}
 where $\V_\Phi=L(\theta)\Phi$. For a curve $\gamma\in C^{\infty}([0, T]; \sP_+)$, its arc-length in manifold $(\Delta_3, \g)$ satisfies 
 \begin{equation*}
 \begin{split}
 \bar L(\gamma)=&\int_0^T \Big(\dot\gamma(t)^{\ts}R(\gamma(t))\gamma(t)\Big)^{\frac{1}{2}}dt\\
=&\int_0^T \Big\{(\nabla_\omega\Phi(t))^2_{12} \theta_{12}(t)+(\nabla_\omega \Phi(t))_{23}^2\theta_{23}(t)+(\nabla_\omega\Phi(t))_{13}^2\theta_{13}(t)\Big\}^{\frac{1}{2}}dt,
\end{split}
 \end{equation*} 
 where $R(t)=R(\theta(\gamma(t)))$, $\theta(t)=\theta(\gamma(t))$, and $\dot\gamma(t)=\V_{\Phi(t)}=L(\theta(\gamma(t)))\Phi(t)$. 
 
 We are ready to derive the Levi-Civita connection. Given vectors $\Phi_1$, $\Phi_2$, $\Phi_3\in\mathbb{R}^3/\mathbb{R}$ and $p\in \Delta_3$, we define 
\begin{equation*}
\begin{split}
\Gamma(\Phi_1, \Phi_2, p)=&\Big((\nabla_\omega\Phi_1)_{12}^2\frac{\partial \theta_{12}}{\partial p_i}+(\nabla_\omega\Phi_1)_{13}^2\frac{\partial\theta_{13}}{\partial p_i}+(\nabla_\omega\Phi_1)_{23}^2\frac{\partial\theta_{23}}{\partial p_i}\Big)_{i=1}^3. 
\end{split}
\end{equation*}
Similarly, one can define $\Gamma(\Phi_1, \Phi_3, p), \Gamma(\Phi_2, \Phi_3, p)\in \mathbb{R}^3$. Then the commutator in $(\Delta_3, \g)$ satisfies  
\begin{equation*}
[\V_{\Phi_1}, \V_{\Phi_2}]=L(\V_{\Phi_1}\theta)\Phi_2-L(V_{\Phi_2}\theta)\Phi_1,
\end{equation*}
with $k=1,2$, 
\begin{equation*}
L(\V_{\Phi_k}\theta)=\begin{pmatrix}
\omega_{12}(\V_{\Phi_k}\theta)_{12}+\omega_{13}(\V_{\Phi_k}\theta)_{13} & -\omega_{12}(\V_{\Phi_k}\theta)_{12} & -\omega_{13}(\V_{\Phi_k}\theta)_{13} \\
-\omega_{12}(\V_{\Phi_k}\theta)_{12}   & \omega_{12}(\V_{\Phi_k}\theta)_{12}+\omega_{23}(\V_{\Phi_k}\theta)_{23} & -\omega_{23}(\V_{\Phi_k}\theta)_{23}\\
-\omega_{13}(\V_{\Phi_k}\theta)_{13}& -\omega_{23}(\V_{\Phi_k}\theta)_{23}   & \omega_{13}(\V_{\Phi_k}\theta)_{13}+\omega_{23}(\V_{\Phi_k}\theta)_{23}
\end{pmatrix},
\end{equation*}
and $i,j=1, 2, 3$, 
\begin{equation*}
(\V_{\Phi_k}\theta)_{ij}=\frac{\partial\theta_{ij}}{\partial p_i}(\V_{\Phi_k})_i+\frac{\partial\theta_{ij}}{\partial p_j}(\V_{\Phi_k})_j. 
\end{equation*}
Thus, the Levi-Civita connection in $(\Delta_3, \g)$ satisfies 
\begin{equation*}
\begin{split}
&\langle \bar\nabla_{\V_{\Phi_1}}\V_{\Phi_2}, \V_{\Phi_3}\rangle\\=&\frac{\theta_{12}}{2} \Big\{(\nabla_\omega \Phi_1)_{12}(\nabla_\omega \Gamma(\Phi_2, \Phi_3, p))_{12}+(\nabla_\omega \Phi_2)_{12}(\nabla_\omega \Gamma(\Phi_1, \Phi_3,p))_{12}\\
&\qquad-(\nabla_\omega \Phi_3)_{12}(\nabla_\omega \Gamma(\Phi_1, \Phi_2,p))_{12}\Big\}\\
+&\frac{\theta_{23}}{2} \Big\{(\nabla_\omega \Phi_1)_{23}(\nabla_\omega \Gamma(\Phi_2, \Phi_3,p))_{23}+(\nabla_\omega \Phi_2)_{23}(\nabla_\omega \Gamma(\Phi_1, \Phi_3,p))_{23}\\
&\qquad-(\nabla_\omega \Phi_3)_{23}(\nabla_\omega \Gamma(\Phi_1, \Phi_2,p))_{23}\Big\}\\
+&\frac{\theta_{13}}{2} \Big\{(\nabla_\omega \Phi_1)_{13}(\nabla_\omega \Gamma(\Phi_2, \Phi_3,p))_{13}+(\nabla_\omega \Phi_2)_{13}(\nabla_\omega \Gamma(\Phi_1, \Phi_3,p))_{13}\\
&\qquad-(\nabla_\omega \Phi_3)_{13}(\nabla_\omega \Gamma(\Phi_1, \Phi_2,p))_{13}\Big\}.
\end{split}
\end{equation*}
\section{Example II: Curvatures of probability manifolds on a three-point lattice graph}\label{sec5}
In this section, we derive analytical expressions for sectional, Ricci, and scalar curvatures of probability manifolds on a three-point graph.  

Consider a three-point graph below. 
\begin{center}
\begin{tikzpicture}[->,shorten >=1pt,auto,node distance=2cm,
        thick,main node/.style={circle,fill=blue!20,draw,minimum size=0.5cm,inner sep=0pt]} ]
   \node[main node] (1) {A};
    \node[main node] (2) [right of=1]  {B};
    \node[main node](3)[right of=2]{C};
    \path[<->]
    (1) edge node {} (2)
    (2) edge node{} (3);
\end{tikzpicture}
\end{center}
Again, denote the probability simplex set as $\Delta_3$. We simply notations that $\theta_1(p):=\theta_{12}(p)$, and $\theta_2(p):=\theta_{23}(p)$. Given a vector $\Phi\in \mathbb{R}^3$ and a point $p\in \Delta_3$, consider the metric $\g$ satisfying 
\begin{equation*}
\begin{split}
\langle \V_\Phi, \V_\Phi\rangle 
=&(\nabla_\omega\Phi)^2_{12} \theta_{1}+(\nabla_\omega \Phi)_{23}^2\theta_{2},
\end{split}
\end{equation*}
 where $\V_\Phi=L(\theta)\Phi$. We let $\omega_{12}=\omega_{23}=1$, $\omega_{13}=0$, such that  
\begin{equation*}
L(\theta)=\begin{pmatrix}
\theta_{12} & -\theta_{12} & 0 \\
-\theta_{12}   & \theta_{12}+\theta_{23} & -\theta_{23}\\
0 & -\theta_{23}   & \theta_{23}
\end{pmatrix}.
\end{equation*}
 There is a particular coordinate for $\Delta_3$, which simplifies geometric calculations. Denote the cumulative distribution function (CDF) on discrete states, such that 
\begin{equation*}
x_1=p_1, \quad x_2=p_1+p_2. 
\end{equation*}
Denote a set $\mathrm{CDF}=\{(x_1, x_2)\in [0,1]^2\colon x_1\leq x_2\}$. Thus, the metric $\g$ in coordinates $(x_1, x_2)$ is a diagonal matrix, such that $\g=(\g_{ij})_{1\leq i,j\leq 2}\in\mathbb{R}^{2\times 2}$, with 
\begin{equation*}
\g_{11}=\frac{1}{\theta_1(x)},\quad \g_{22}=\frac{1}{\theta_2(x)}, \quad \g_{12}=\g_{21}=0. 
\end{equation*}
We next derive formulas for Riemannian curvatures. 
\begin{proposition}[Curvatures on $\Delta_3$ with a lattice graph]
The sectional curvature satisfies 
\begin{equation}\label{K12}
\bar K_{12}(p)=\frac{1}{\theta_2}\Big[\frac{1}{2}\partial_{11}\log \theta_2+\frac{1}{4}\partial_1\log\frac{\theta_1}{\theta_2}\cdot \partial_1\log\theta_2\Big]+\frac{1}{\theta_1}\Big[\frac{1}{2}\partial_{22}\log\theta_1+\frac{1}{4}\partial_2\log\frac{\theta_2}{\theta_1}\cdot\partial_2\log\theta_1\Big].
\end{equation}
The Ricci curvature satisfies 
\begin{equation*}
\bar R_{11}(p)=\bar K_{12}(p)\theta_2(p), \quad \bar R_{22}(p)=\bar K_{12}(p)\theta_1(p), \quad \bar R_{12}(p)=\bar R_{21}(p)=0. 
\end{equation*}
The scalar curvature satisfies
\begin{equation*}
\bar S(p)=2\bar K_{12}(p)\theta_1(p) \theta_2(p).
\end{equation*}
\end{proposition}
\begin{proof}
The proof is based on standard geometric calculations for curvatures with a diagonal metric. The Riemannian curvature tensor $\bar R$ in $(\mathrm{CDF}, \g)$ is computed as follows. Let $i,j,k,l\in\{1,2\}$. Denote the Levi-Civita connection as
\begin{equation*}
\bar\Gamma_{ij}^k=\frac{1}{2}\sum_{k=1}^2 \theta_k\Big(\partial_i\g_{jk}+\partial_j\g_{ik}-\partial_k\g_{ij}\Big),
\end{equation*}
for $i,j,k=1,2$. Thus, using the fact $\partial_i\log\theta_j=\frac{\partial_i\theta_j}{\theta_j}$, for $i,j=1,2$, we have
\begin{equation*}
\left\{\begin{aligned}
&\bar\Gamma_{11}^1=-\frac{1}{2}\partial_1\log\theta_1,\quad \bar\Gamma_{11}^2=\frac{1}{2}\frac{\theta_2}{\theta_1}\partial_2\log\theta_1,\\
&\bar\Gamma_{22}^1=\frac{1}{2}\frac{\theta_1}{\theta_2}\partial_1\log\theta_2, \quad \bar\Gamma^2_{22}=-\partial_2\log\theta_2,\\
&\bar\Gamma_{12}^1=-\frac{1}{2}\partial_2\log\theta_1,\quad \bar\Gamma_{12}^2=-\frac{1}{2}\partial_1\log\theta_2. 
\end{aligned}\right.
\end{equation*}
We note that the Riemannian curvature satisfies 
\begin{equation*}
\bar R^i_{jkl}=\partial_k\bar\Gamma_{lj}^i-\partial_l\bar\Gamma^i_{kj}+\sum_{m=1}^2\Big(\bar\Gamma_{lj}^m\cdot\bar\Gamma_{km}^i-\bar\Gamma_{kj}^m\cdot\bar\Gamma_{lm}^i\Big), 
\end{equation*}
for $i,j,k,l=1,2$. Thus, by direct calculations, we have $\bar R_{111}^1=\bar R^1_{211}=\bar R_{221}^2=\bar R_{222}^2=0$, and 
\begin{equation*}
\left\{\begin{aligned}
&\bar R_{121}^2=\frac{1}{2}\partial_{11}\log\theta_2+\frac{1}{4}\partial_1\log\theta_2\cdot\partial_1\log\frac{\theta_1}{\theta_2}+\frac{\theta_2}{\theta_1}\Big[\frac{1}{2}\partial_{22}\log\theta_1+\frac{1}{4}\partial_2\log\theta_1\cdot\partial_2\log\frac{\theta_2}{\theta_1}\Big],\\
&\bar R_{212}^1=\frac{1}{2}\partial_{22}\log\theta_1+\frac{1}{4}\partial_2\log\theta_1\cdot\partial_2\log\frac{\theta_2}{\theta_1}+\frac{\theta_1}{\theta_2}\Big[\frac{1}{2}\partial_{11}\log\theta_2+\frac{1}{4}\partial_1\log\theta_2\cdot\partial_1\log\frac{\theta_1}{\theta_2}\Big].
\end{aligned}\right.
\end{equation*}
We note the following facts from the definition of curvatures. The section curvature is defined by $\bar K_{12}=\frac{1}{\theta_2}\bar R_{121}^2$. The Ricci tensor is defined by $\bar R_{jk}=\sum_{i=1}^2\bar R^i_{kij}$. And the scalar curvature is defined by $\bar S=\sum_{i=1}^2\theta_i\bar R_{ii}$. We then exchange the coordinates from the cumulative distribution function $x$ to the probability function $p$. This finishes the derivation. 
\end{proof}
We also derive Ricci curvatures for a class of theta functions from Onsager's response matrix with the $\alpha$-divergence mean, and the geometric mean.   

\begin{example}[Curvatures from Onsager's response matrix with $\alpha$-divergence means]\label{ex2}
Denote $\theta_i(p)=c\cdot\frac{p_i-p_{i+1}}{f'(cp_i)-f'(cp_{i+1})}$, with $c=3$. Let $p_1$, $p_2$, $p_3$ be distinct numbers. Then the sectional curvature satisfies 
\begin{equation*}
\begin{split}
\bar K_{12}(p)=&-\Big\{\quad\frac{1}{2(p_1-p_2)^2}\Big[\frac{3}{2\theta_1}-\frac{1}{2}f''(cp_2)^2\theta_1-\big(f''(cp_2)-cf'''(cp_2)(p_2-p_1)\big)\Big]\\
&\qquad+\frac{1}{2(p_2-p_3)^2}\Big[\frac{3}{2\theta_2}-\frac{1}{2}f''(cp_2)^2\theta_2-\big(f''(cp_2)-cf'''(cp_2)(p_2-p_3)\big)\Big]\\
&\qquad+\frac{1}{4(p_2-p_1)(p_2-p_3)}\big(2-(f''(cp_2)+f''(cp_3))\theta_2\big)(f''(cp_2)-\frac{1}{\theta_1})\\
&\qquad+\frac{1}{4(p_2-p_1)(p_2-p_3)}\big(2-(f''(cp_1)+f''(cp_2))\theta_1\big)(f''(cp_2)-\frac{1}{\theta_2})\Big\}.
\end{split}
\end{equation*}
\begin{itemize}
\item[(i)] Consider the KL divergence with $f(z)=z\log z-(z-1)$ and the logarithm mean $\theta_i(p)=c\cdot \frac{p_i-p_{i+1}}{\log p_i-\log p_{i+1}}$. The sectional curvature satisfies 
\begin{equation*}
\begin{split}
\bar K_{12}(p)=&-\Big\{\quad\frac{1}{2(p_1-p_2)^2}\Big[\frac{3}{2\theta_1}+\frac{1}{c}\big(\frac{p_1}{p_2}-2)\frac{1}{p_2}-\frac{\theta_1}{2c^2p_2^2}\Big]\\
&\qquad+\frac{1}{2(p_2-p_3)^2}\Big[\frac{3}{2\theta_2}+\frac{1}{c}(\frac{p_3}{p_2}-2)\frac{\theta_2}{p_2}-\frac{\theta_2^2}{2c^2p_2^2}\Big]\\
&\qquad+\frac{1}{4(p_2-p_1)(p_2-p_3)}\big(2-\frac{1}{c}(\frac{1}{p_2}+\frac{1}{p_3})\theta_2\big)(\frac{1}{cp_2}-\frac{1}{\theta_1})\\
&\qquad+\frac{1}{4(p_2-p_1)(p_2-p_3)}\big(2-\frac{1}{c}(\frac{1}{p_1}+\frac{1}{p_2})\theta_1\big)(\frac{1}{cp_2}-\frac{1}{\theta_2})\Big\}.
\end{split}
\end{equation*}
\item[(ii)]  Consider the $\alpha$-divergence with $f(z)=\frac{4}{1-\alpha^2}(\frac{1-\alpha}{2}+\frac{1+\alpha}{2}z-z^{\frac{1+\alpha}{2}})$, $\alpha\neq 1$, and the polynomial mean $\theta_i=\frac{1}{2}c^{\frac{3-\alpha}{2}}(\alpha-1)\cdot \frac{p_i-p_{i+1}}{p_i^{\frac{\alpha-1}{2}}-p_{i+1}^{\frac{\alpha-1}{2}}}$. The sectional curvature satisfies 
\begin{equation*}
\begin{split}
\bar K_{12}(p)=&-\Big\{\quad\frac{1}{2(p_1-p_2)^2}\Big[\frac{3}{2\theta_1}-\frac{1}{2}(cp_2)^{\alpha-3}\theta_1-\big((cp_2)^{\frac{\alpha-3}{2}}-\frac{c(\alpha-3)}{2}(cp_2)^{\frac{\alpha-5}{2}}(p_2-p_1)\big)\Big]\\
&\qquad+\frac{1}{2(p_2-p_3)^2}\Big[\frac{3}{2\theta_2}-\frac{1}{2}(cp_2)^{\alpha-3}\theta_2-\big((cp_2)^{\frac{\alpha-3}{2}}-\frac{c(\alpha-3)}{2}(cp_2)^{\frac{\alpha-5}{2}}(p_2-p_3)\big)\Big]\\
&\qquad+\frac{1}{4(p_2-p_1)(p_2-p_3)}\big(2-((cp_2)^{\frac{\alpha-3}{2}}+(cp_3)^{\frac{\alpha-3}{2}})\theta_2\big)((cp_2)^{\frac{\alpha-3}{2}}-\frac{1}{\theta_1})\\
&\qquad+\frac{1}{4(p_2-p_1)(p_2-p_3)}\big(2-((cp_1)^{\frac{\alpha-3}{2}}+(cp_2)^{\frac{\alpha-3}{2}})\theta_1\big)((cp_2)^{\frac{\alpha-3}{2}}-\frac{1}{\theta_2})\Big\}.
\end{split}
\end{equation*}
\end{itemize}
From numerics in Figure \ref{fig1}, we observe that the sectional curvature with logarithm mean in probability manifolds is always negative. 

\begin{figure}
    \centering
    \includegraphics[width=0.45\linewidth]{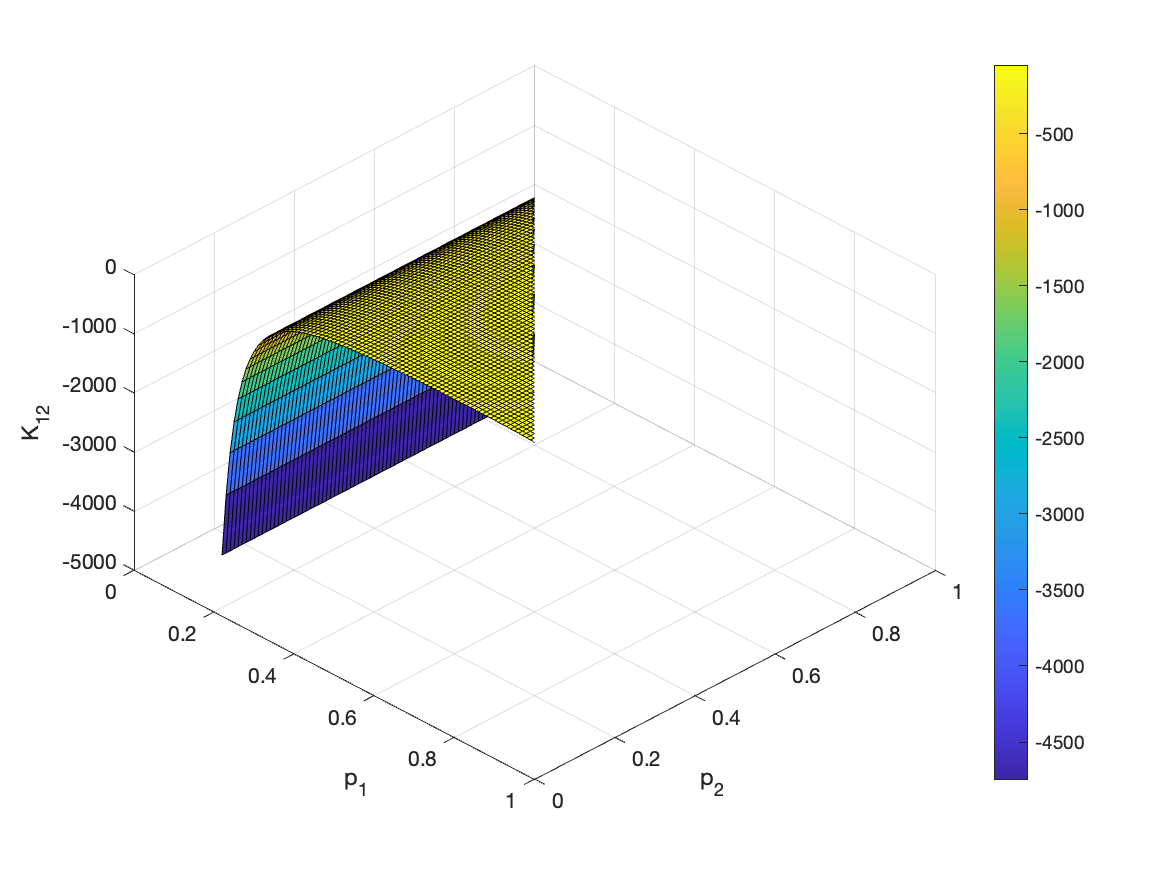}
        \includegraphics[width=0.45\linewidth]{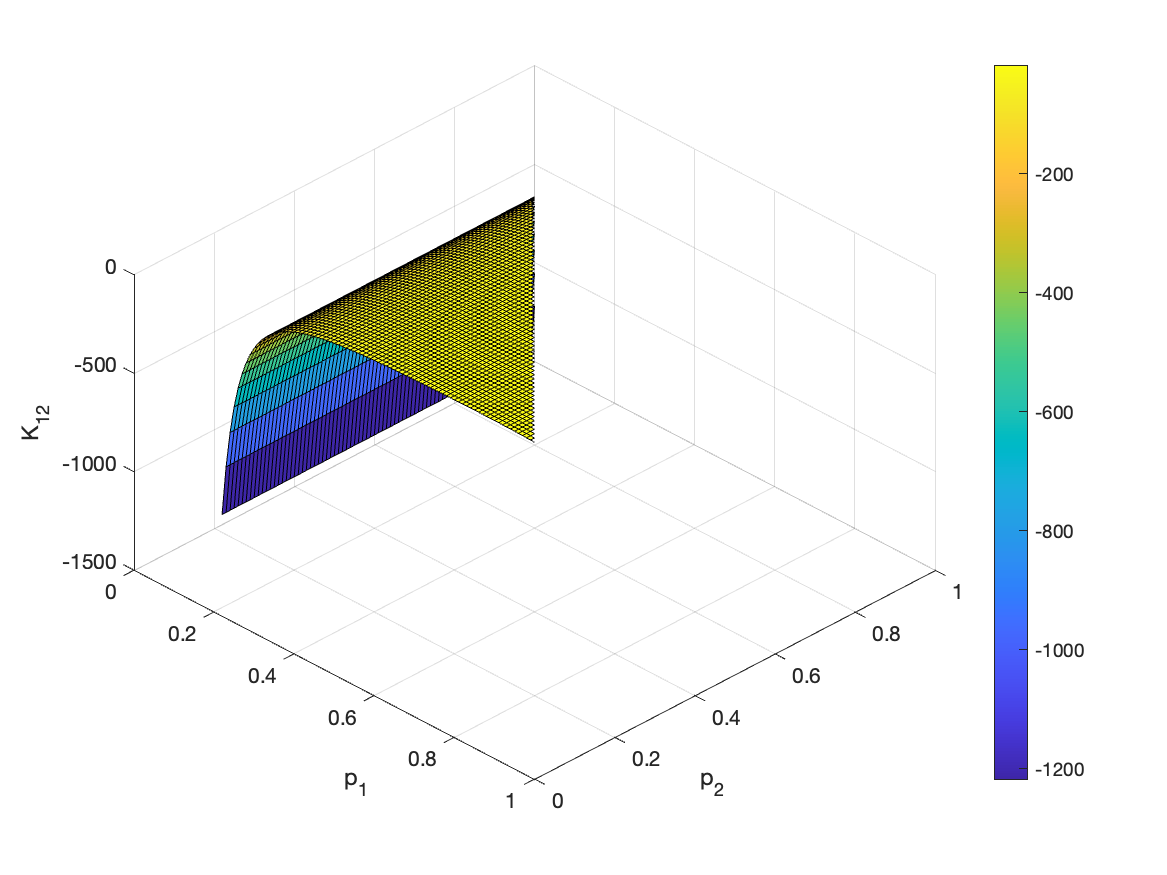}\\
            \includegraphics[width=0.45\linewidth]{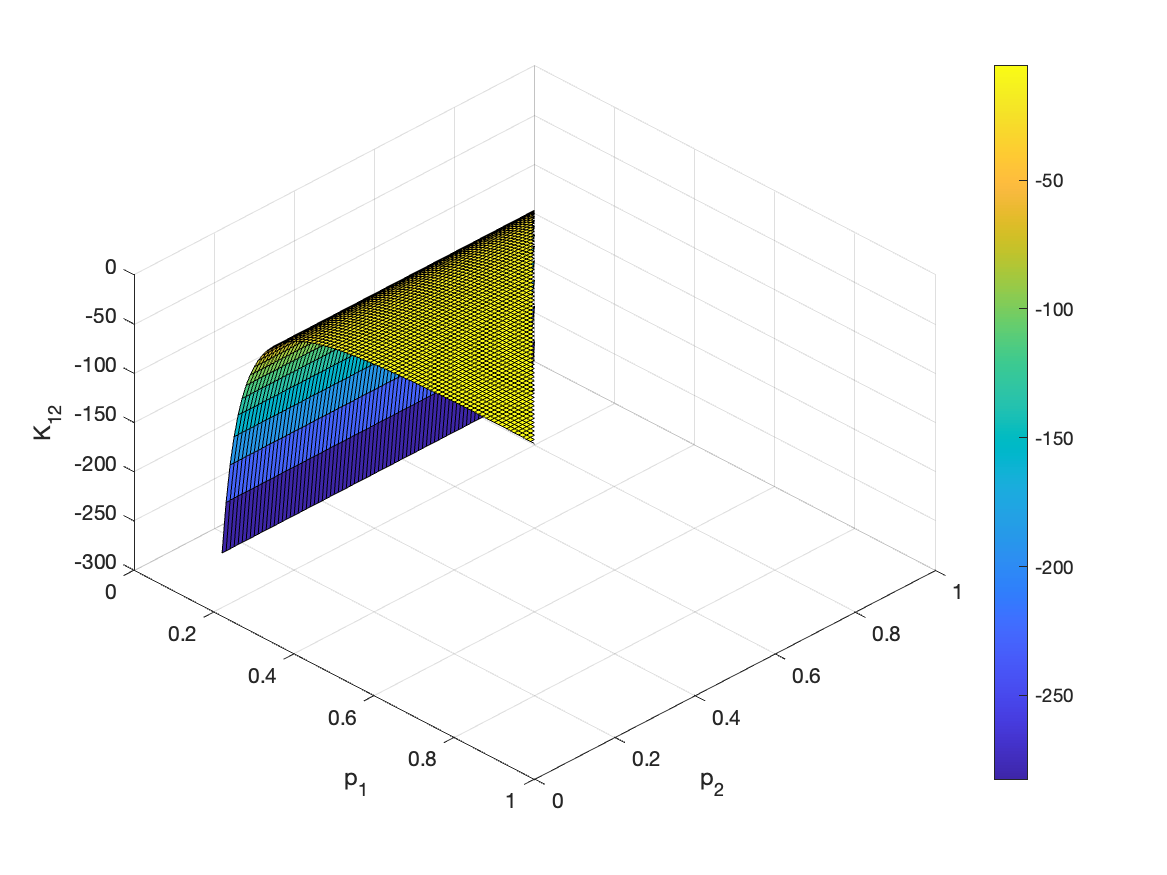}
                \includegraphics[width=0.45\linewidth]{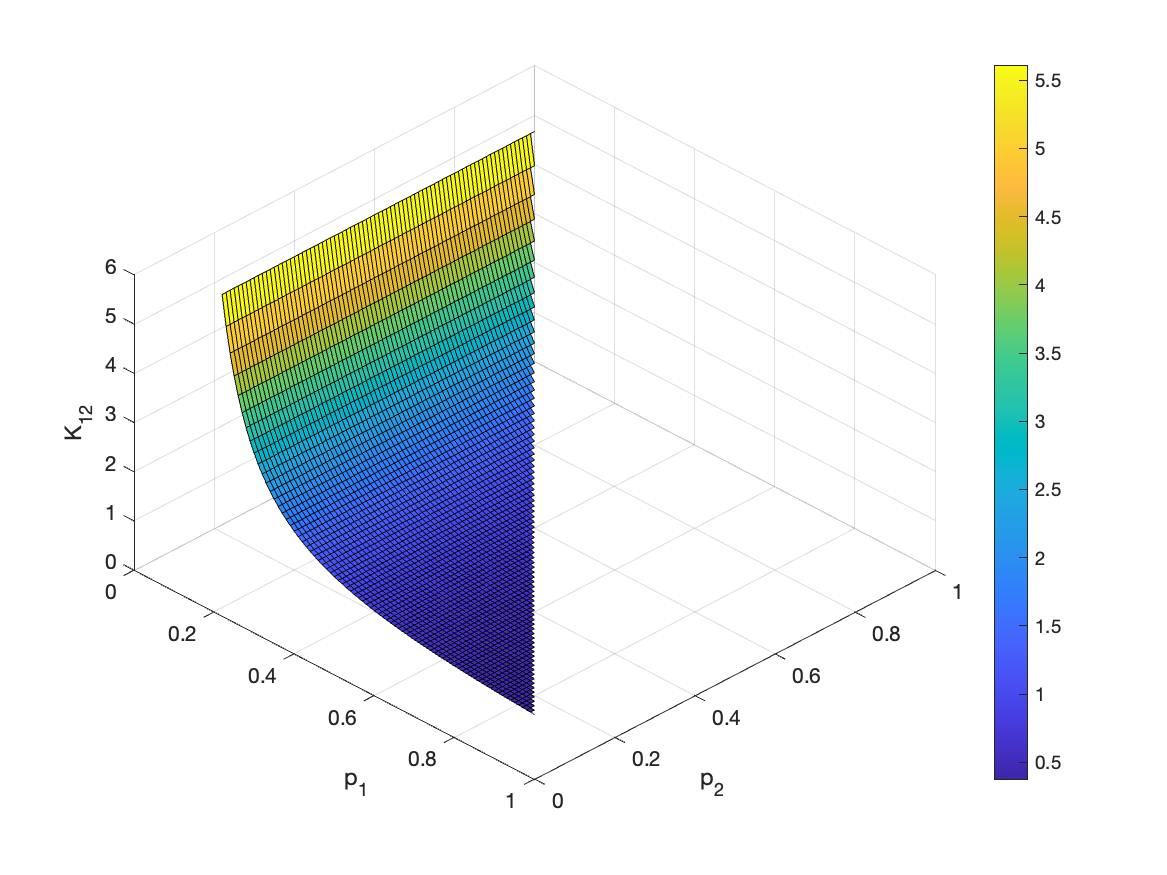}\\
   \caption{ \label{fig1}Plots of sectional curvatures with $\alpha$-divergence means in Example \ref{ex2} for $\alpha=-1$ (up left), $\alpha=0$ (up right), $\alpha=1$ (low left), and $\alpha=4$ (low right).}
\end{figure}
\end{example}
\begin{example}[Curvatures from Onsager's response matrix with generalized geometric means]\label{ex3}
Denote $\theta_i(p)=c\cdot p_i^\beta p_{i+1}^\beta$, with $\beta\in\mathbb{R}$ and $c=3^{2\beta}>0$. 
If $\beta=\frac{1}{2}$, $\theta$ is the geometric mean function. 
Then the sectional curvature satisfies 
\begin{equation*}
\bar K_{12}(p)=-\frac{1}{2}\Big[\frac{1}{\theta_2}\big(\frac{\beta}{p_2^2}+\frac{\beta^2}{2p_1p_2}\big)+\frac{1}{\theta_1} \big(\frac{\beta}{p_2^2}+\frac{\beta^2}{2p_2p_3}\big)\Big]. 
\end{equation*}
The Ricci curvature satisfies 
\begin{equation*}
\bar R_{11}(p)=-\frac{1}{2}\Big[\frac{\beta}{p_2^2}+\frac{\beta^2}{2p_1p_2}+(\frac{p_3}{p_1})^\beta \big(\frac{\beta}{p_2^2}+\frac{\beta^2}{2p_2p_3}\big)\Big],
\end{equation*}
and 
\begin{equation*}
\bar R_{22}(p)=-\frac{1}{2}\Big[\frac{\beta}{p_2^2}+\frac{\beta^2}{2p_2p_3}+(\frac{p_1}{p_3})^\beta (\frac{\beta}{p_2^2}+\frac{\beta^2}{2p_1p_2})\Big],
\end{equation*}
with $\bar R_{12}(p)=0$. In addition, the scalar curvature satisfies 
\begin{equation*}
\bar S(p)=-c\beta\Big[p_1^\beta p_2^{\beta-2}+p_2^{\beta-2}p_3^\beta+\frac{\beta}{2}(p_1^{\beta-1}p_2^{\beta-1}+p_2^{\beta-1}p_3^{\beta-1})\Big].
\end{equation*}
If $\beta>0$, the sectional curvature in probability manifold $(\mathcal{P}_+, \g)$ is always negative. And the Ricci curvature is always a negative definite matrix. We also provide several numerical examples of sectional curvatures in Figure \ref{fig2}. 
\begin{figure}
    \centering
    \includegraphics[width=0.45\linewidth]{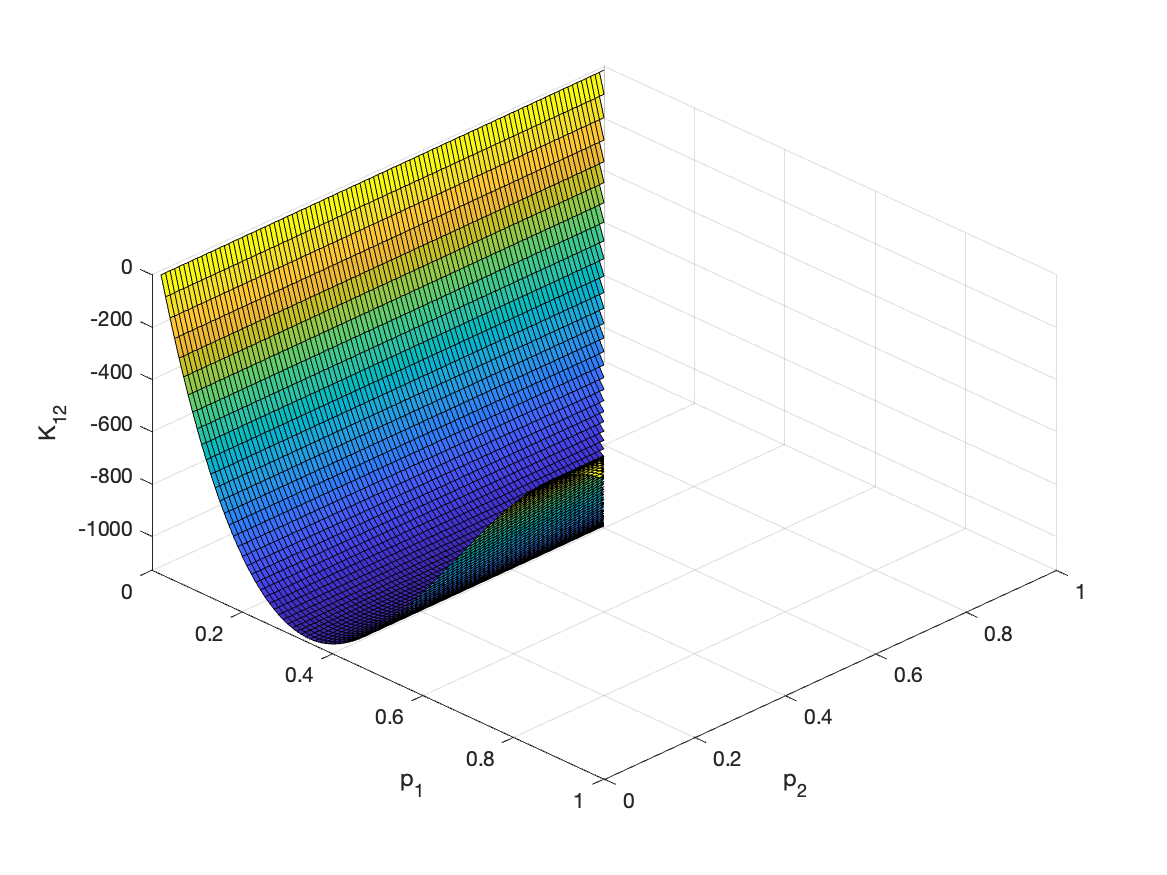}
        \includegraphics[width=0.45\linewidth]{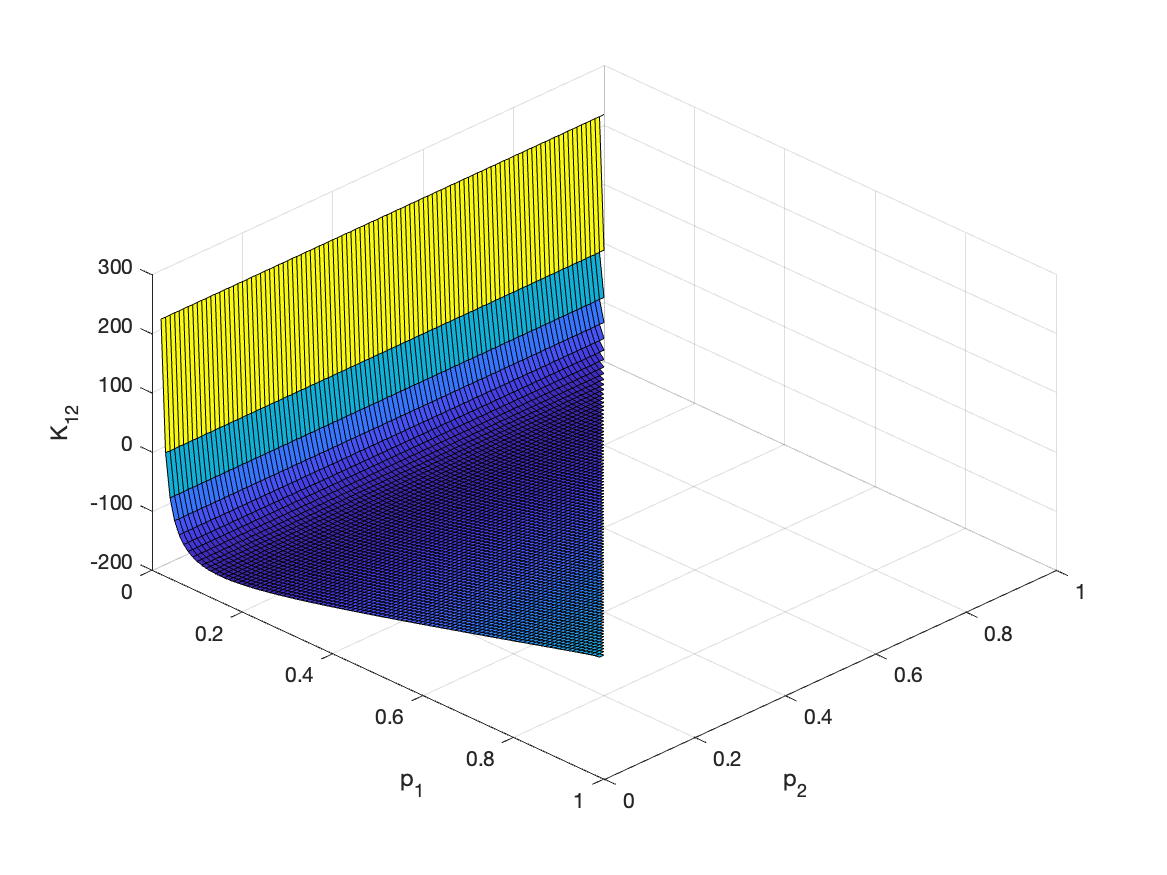}\\
            \includegraphics[width=0.45\linewidth]{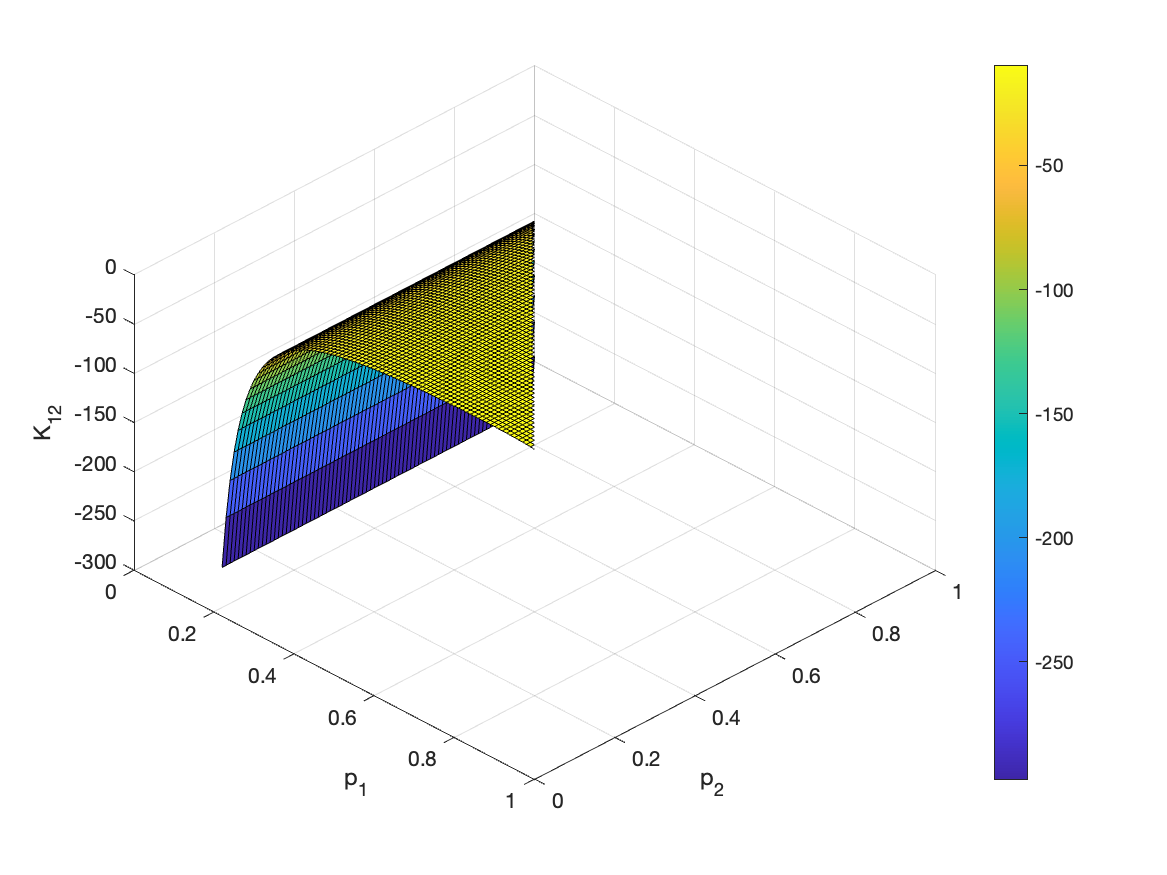}
                \includegraphics[width=0.45\linewidth]{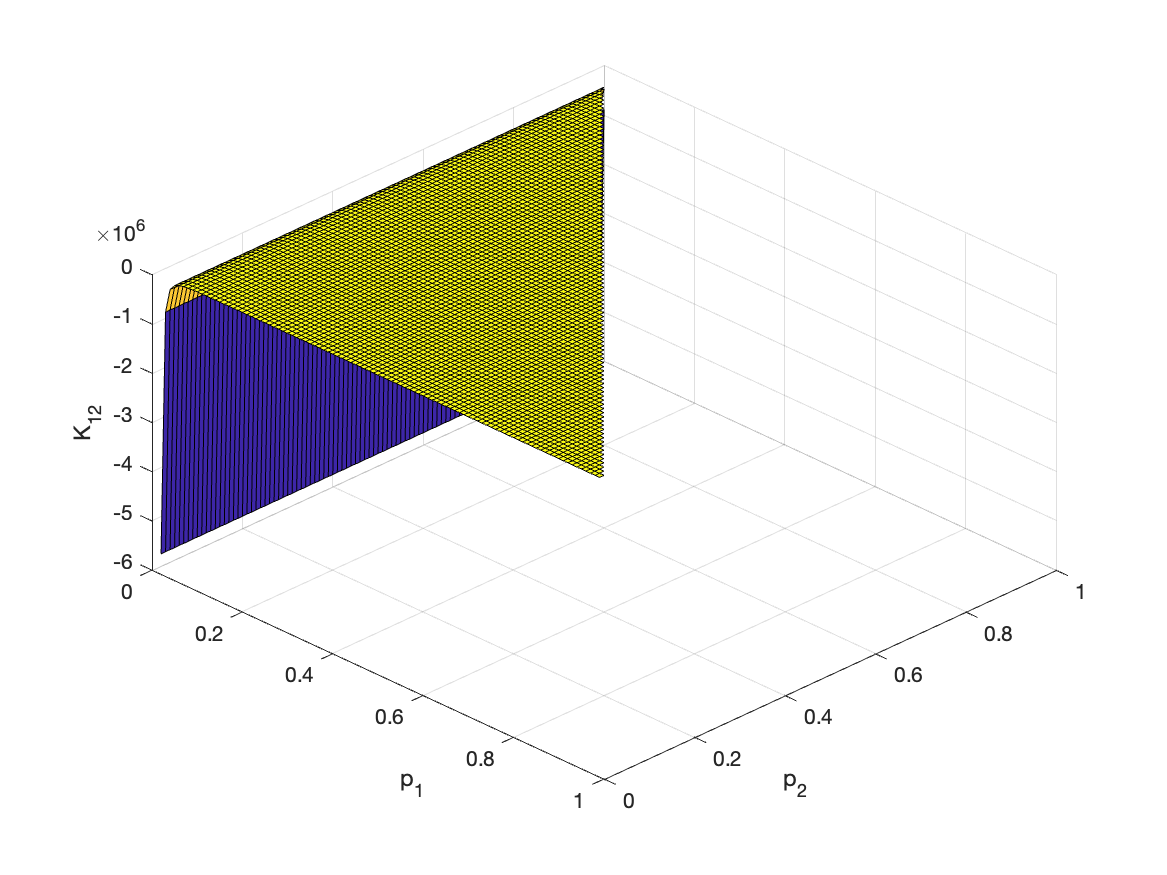}\\
   \caption{\label{fig2}Plots of sectional curvatures with generalized geometric means in Example \ref{ex3} for $\beta=-2$ (up left), $\beta=-1$ (up right), $\beta=\frac{1}{2}$ (low left), and $\beta=1$ (low right).}
\end{figure}

\end{example}
In the appendix, we present the proofs of Examples \ref{ex2} and \ref{ex3}. 

 \section{Discussions}
This paper presents geometric calculations in probability manifolds induced from Onsager reciprocal relations. In other words, irreversible processes in complex systems can be formulated as gradient flows in the probability manifold $(\mathcal{P}_+, \g)$, where Onsager's response matrix constructs the metric $\g$. This paper derives the Riemannian Levi-Civita connection, parallel transport, geodesics, Riemannian curvature, and sectional curvature in probability manifolds. 

There are insights and applications from geometric calculations in probability manifolds for studying ``geometry'' on discrete domains. One example is the generalization of Gamma calculus \cite{BE}. Typical examples include the iterative Gamma and Gamma-2 operators. We derive a generalization towards the classical Gamma calculus from the Hessian operator in the thermodynamical probability manifold; see formula \eqref{Hess_explicit}. This paper also proposes a generalized Gamma calculus in formula \eqref{tensor_explicit}. We shall investigate properties of Riemannian curvatures, in particular Ricci curvatures, in probability manifolds, which characterize stochastic dynamical behaviors of complex systems \cite{LiN1}. 

We also mention that there are generalized Onsager response matrices from nonlinear chemical master equations \cite{Onsager}. They are related to general choices of $f$-divergence functions and physical systems in generalized Onsager reciprocal relations. Typical examples include reaction-diffusion equations on discrete domains and dynamical density functional theories on lattices \cite{DDFT}. Systemic computations and spectral analysis from geometric calculations in probability manifolds are essential in understanding dynamical behaviors and fluctuations of thermodynamics. In addition, there are recent studies in non-equilibrium dynamics, from Macroscopic fluctuation theory (MFT) \cite{MFT} and General Equation for Non-Equilibrium Reversible-Irreversible Coupling (GENERIC) \cite{GEN}. These studies formulate additional nonreciprocal structures within the Onsager principles, including the conservation of energy and the dissipation of Lyapunov functions \cite{Ito}.
In future works, we leave the geometric calculations and analysis for non-reversible or nonreciprocal relation systems. We shall apply Riemannian curvature on probability manifolds to analyze and design numerical algorithms for these systems. 

\noindent\textbf{Acknowledgements}. {W. Li's work is supported by AFOSR YIP award No. FA9550-23-1-0087, NSF RTG: 2038080, NSF DMS-2245097, and the McCausland Faculty Fellow in University of South Carolina.}

\section*{Appendix}
In this section, we present the proofs of Examples \ref{ex2} and \ref{ex3}. 

\begin{proof}[Proof of Example \ref{ex2}]
Note that $p_1=x_1$, $p_2=x_2-x_1$, $p_3=1-x_2$. We have 
\begin{equation*}
\theta_1=c\cdot\frac{2x_1-x_2}{f'(cx_1)-f'(c(x_2-x_1))}, \quad \theta_2=c\cdot\frac{2x_2-x_1-1}{f'(c(x_2-x_1))-f'(c(1-x_2))}.
\end{equation*}
Thus, we have
\begin{equation*}
\begin{split}
\partial_1\log\theta_2=&-\frac{1}{2x_2-x_1-1}+\frac{cf''(c(x_2-x_1))}{f'(c(x_2-x_1))-f'(c(1-x_2))}\\
=&-\frac{1}{p_2-p_3}+\frac{cf''(cp_2)}{f'(cp_2)-f'(cp_3)}.
\end{split}
\end{equation*}
In addition, 
\begin{equation*}
\begin{split}
\partial_{11}\log\theta_2=&-\frac{1}{(2x_2-x_1-1)^2}+\frac{c^2f''(c(x_2-x_1))^2}{(f'(c(x_2-x_1))-f'(c(1-x_2)))^2}\\
&-\frac{c^2f'''(c(x_2-x_1))}{f'(c(x_2-x_1))-f'(c(1-x_2))}\\
=&-\frac{1}{(p_2-p_3)^2}+\frac{c^2f''(cp_2)^2}{(f'(cp_2)-f'(cp_3))^2}-\frac{c^2f'''(cp_2)}{f'(cp_2)-f'(cp_3)}.
\end{split}
\end{equation*}
And
\begin{equation*}
\begin{split}
\partial_1\log\theta_1=&\frac{2}{2x_1-x_2}-\frac{cf''(cx_1)+cf''(c(x_2-x_1))}{f'(cx_1)-f'(c(x_2-x_1))}\\
=&\frac{2}{p_1-p_2}-c\cdot\frac{f''(cp_1)+f''(cp_2)}{f'(cp_1)-f'(cp_2)}.
\end{split}
\end{equation*}
Similarly, we derive the formulas for $\partial_2\log\theta_1$, $\partial_{22}\log\theta_1$, $\partial_2\log\theta_2$. From equation \eqref{K12}, we finish the proof. 
\end{proof}
\begin{proof}[Proof of Example \ref{ex3}]
Note that $\theta_1=c(x_2-x_1)^\beta x_1^\beta$, $\theta_2=c(x_2-x_1)^\beta(1-x_2)^\beta$. 
Thus, we have
\begin{equation*}
\left\{\begin{aligned}
&\partial_1\log\theta_1=\beta(\frac{1}{x_1}-\frac{1}{x_2-x_1}),  \quad\partial_1\log\theta_2=-\frac{\beta}{x_2-x_1},\\
&\partial_{2}\log\theta_1=\frac{\beta}{x_2-x_1},  \quad\partial_2\log\theta_2=\beta(\frac{1}{x_2-x_1}-\frac{1}{1-x_2}),\\
&\partial_{22}\log\theta_1=-\frac{\beta}{(x_2-x_1)^2}, \quad \partial_{11}\log\theta_2=-\frac{\beta}{(x_2-x_1)^2}.
\end{aligned}\right.
\end{equation*}
From equation \eqref{K12}, we finish the proof. 
\end{proof}

\end{document}